\documentclass[letterpaper]{article} 
\usepackage{aaai24}  
\usepackage{times}  
\usepackage{helvet}  
\usepackage{courier}  
\usepackage[hyphens]{url}  
\usepackage{graphicx} 
\usepackage{natbib}  
\usepackage{caption} 
\usepackage{algorithm}
\usepackage{algorithmic}

\usepackage{newfloat}
\usepackage{listings}
\DeclareCaptionStyle{ruled}{labelfont=normalfont,labelsep=colon,strut=off} 
\floatstyle{ruled}
\newfloat{listing}{tb}{lst}{}
\floatname{listing}{Listing}
\nocopyright

\usepackage{amsmath}
\usepackage{amssymb}
\usepackage{amsthm}
\usepackage{caption}
\usepackage{subcaption}
\usepackage{comment}
\specialcomment{prewebconf}{\color{blue}}{\color{black}}
\excludecomment{prewebconf}

\usepackage{graphicx}
\usepackage{sgame}
\usepackage{color}
\usepackage{enumitem}
\usepackage{url}

\usepackage{color-edits} 
\addauthor{as}{red}     
\addauthor{jg}{blue}    
\addauthor{jb}{magenta} 
\addauthor{bl}{orange}  
\addauthor{ni}{magenta} 
\addauthor{ia}{magenta} 

\usepackage{etoolbox}
\usepackage{chngcntr}
\usepackage{multirow}
\usepackage{apptools}
\newtheorem{theorem}{Theorem}[section]
\newtheorem{prop}[theorem]{Proposition}
\newtheorem{proposition}[theorem]{Proposition}

\newtheorem{lemma}[theorem]{Lemma}

\newtheorem{corollary}[theorem]{Corollary}
\newtheorem{rem}[theorem]{Remark}
\theoremstyle{definition}

\usepackage{hhline}
\renewcommand*{\big}{\bigg}

\usepackage{slivkins-setup} 
\renewcommand{\fakeItem}[1][$\bullet$]{\vspace{1mm}\noindent{\textbf{#1}}~}

\usepackage{xspace,bm,upgreek} 
\usepackage{diagbox,bigstrut,makecell} 
\newcommand{\hlinestrut}{\bigstrut[b]\\\hline\bigstrut[t]}

\newcommand{\emthm}[1]{{\bf\em #1}}  
\newcommand{\term}[1]{\mathtt{#1}\xspace}
\newcommand{\bin}{\{0,1\}}  
\newcommand{\KL}[2]{D_{\mathtt{KL}}\rbr{#1\,\|\,#2}} 

\newcommand{\vot}{VoTC\xspace}
\newcommand{\mvot}{MVoT\xspace}

\newcommand{\ourGame}{content-filtering game\xspace} 

\newcommand{\pot}{\tilde{\pi}^*_0}
\newcommand{\ptt}{\tilde{\pi}^*_1}

\newcommand{\fil}{\mathtt{f}} 
\newcommand{\con}{\mathtt{c}} 

\newcommand{\signal}{\Psi}    
\newcommand{\sfil}{\signal_{\fil}} 
\newcommand{\scon}{\signal_{\con}} 

\newcommand{\afil}{a_{\fil}} 
\newcommand{\acon}{a_{\con}} 

\newcommand{\strfil}{s_{\fil}} 
\newcommand{\strcon}{s_{\con}} 

\newcommand{\mixsigma}{\bm{\upsigma}} 

\newcommand{\mixstr}{\mathbf{s}} 
\newcommand{\mixstrfil}{\mixsigma_{\fil}} 
\newcommand{\mixstrcon}{\mixsigma_{\con}} 

\newcommand{\Vfil}{V_{\fil}} 
\newcommand{\vcon}{v_{\con}} 
\newcommand{\Vcon}{V_{\con}} 
\newcommand{\Vatt}{V_{\term{a}}} 
\newcommand{\Vind}{\widehat{V}} 

\newcommand{\Vopt}{V^*}    
\newcommand{\infoC}{\term{C}}
\newcommand{\infoL}{\lambda}

\newcommand{\fwdprof}{\sigma_{\term{fwd}}}   
\newcommand{\blkprof}{\sigma_{\term{blk}}}   
\newcommand{\difprof}{\sigma_{\term{dif}}}   
\newcommand{\mixprof}{\sigma_{\mathtt{mix}}} 

\newcommand{\strfwd}{s_{\term{fwd}}} 
\newcommand{\strblk}{s_{\term{blk}}} 
\newcommand{\strdif}{s_{\term{dif}}} 

\newcommand{\qdif}{q_{\term{dif}}}  
\newcommand{\qH}{q_{\term{H}}}
\newcommand{\qL}{q_{\term{L}}}
\newcommand{\qfunc}{q} 
\newcommand{\qhat}{\hat{q}}  

\newcommand{\partialpi}[1]{\frac{\partial}{\partial \pi_{#1}}}  
\newcommand{\partialV}[1]{\rbr{\partialpi{0}{#1},\; -\partialpi{1}{#1}}}

\newcommand{\Vnorm}{\widetilde{V}} 
\newcommand{\deltaU}{\vartriangle\!\! U}

\begin{document}

\setcounter{secnumdepth}{1} 

\title{Content Filtering with Inattentive Information Consumers%
\thanks{First version: May 2022. This version: December 2024. \newline
This is the full version of a conference paper published in \emph{AAAI'24}.
\newline
JG and IB were affiliated with Microsoft during this research.}}

\author {
    Ian Ball\textsuperscript{\rm 1},
    James Bono\textsuperscript{\rm 2},
    Justin Grana\textsuperscript{\rm 3},
    Nicole Immorlica\textsuperscript{\rm 4},
    Brendan Lucier\textsuperscript{\rm 4},
    Aleksandrs Slivkins\textsuperscript{\rm 5}
}
\affiliations {
    \textsuperscript{\rm 1}MIT, 77 Massachusetts Ave, Cambridge, MA 02139, USA; ianball@mit.edu.\\
    \textsuperscript{\rm 2}Microsoft, 1 Microsoft Way, Redmond, WA 98052, USA; james.bono@microsoft.com.\\
    \textsuperscript{\rm 3}Edge \& Node, remote only; justin@edgeandnode.com.\\
    \textsuperscript{\rm 4} Microsoft Research, 1 Memorial Dr, Cambridge, MA 02142, USA; \{nicimm, brlucier\}@microsoft.com.\\
    \textsuperscript{\rm 5} Microsoft Research, 300 Lafayette St, New York, NY 10012, USA; slivkins@microsoft.com.
}

\newcommand{\UnreasonableStrategyLocation}{Proposition~\ref{prop:unreason}\xspace}
\newcommand{\GenCostsLocation}{Appendix~\ref{sec:gen}\xspace}

\maketitle
\begin{abstract}
    We develop  
  a model of content filtering as a game between the filter and the content consumer, 
  where the latter incurs information costs for examining the content. 
  Motivating examples 
  include censoring misinformation, spam/phish
  filtering, and recommender systems acting on a stream of content.  
  When the attacker is exogenous,
  we show that improving the filter's quality is weakly Pareto improving,
  but has no impact on equilibrium payoffs until the filter becomes sufficiently accurate. 
  Further, if the filter does not internalize the consumer's information costs,
  its lack of commitment power may render it 
  useless and lead to inefficient outcomes.  
  When the attacker is also strategic,
  improvements in filter quality may   
  decrease equilibrium payoffs. 
  

  \end{abstract}

















\section{Introduction}
\label{sec:intro}
Content filtering is a crucial and widely-applied tool for
improving the experience of information consumers.
Email filters automatically sort normal, malicious and spam messages,
increasing security and saving users from
manually sorting mail
\citep{filt1,filt2,filt3}.
Information aggregators and social
media platforms have deployed content filters that censor non-credible
and potentially deceptive claims \cite{fn1,fn2}.
Recommender systems learn consumers'
preferences to save them from having to sift through unwanted content
\cite{rec1,rec2,rec3}.\footnote{Our model is most relevant to recommender systems that process a stream of items such as new event announcements: \eg new concerts for a music app, or new properties for a real estate app.}
%
Despite major efforts to improve content filters,
information consumers
remain susceptible to malicious or illegitimate content, \eg
they click on phishing messages \cite{pc1,pc2}
and fall victim to misinformation \cite{mis1,mis2}.


Consumers can take measures to avoid the malicious content. For
example, a recipient of a suspicious email could examine the
  email more carefully, do a quick web search for known malicious
  patterns, ask an acquaintance's opinion, or even attempt
  to reach the purported sender by other means. A social media user
could carefully check the argumentation in a given post, or
  consult reputable sources.  However, such measures incur
substantial costs in time, effort and attention. In particular, the
literature on ``attention economy'' documents that attention in the digital
sphere is a scarce resource \cite{attmark}. We will refer
to these costs as \emph{information costs}%
\footnote{An alternative term, \emph{attention costs}, is also
    well-established.}.

Due to information costs, consumers  tend to strategically alter their
  behavior in response to the (perceived) filter quality. When
  consumers perceive that a filter is poor, either allowing too much
  malicious content or censoring too much, they abandon the platform
  (a risk acknowledged by major platforms for email, social media
  and news \citep{cnbc}).  When the filter is exceptional,
  consumers take content at face value \cite{whoshared}. In the
  ``middle ground," the filter is imperfect and consumers choose
  whether/how to examine the content to determine its quality.%
\footnote{An ironic example: a conference
    serves as a filter for academic publications, and its reputation
    (\ie perceived filter quality) is often used to evaluate the merit
    of a scientific claim \cite{journal1}.}



The considerable investment in improving content filters and consumers' strategic allocation of scarce attention  motivates three salient questions:

\fakeItem[(Q1)] Can the benefits of an increase in filter
quality be crowded out by reduced consumer attention in response to the increase in filter quality?

\fakeItem[(Q2)]  If the filter's payoffs do not depend on the consumer's information costs, what inefficiencies (i.e. sub-optimal equilibria) arise and how can they be abated?

\fakeItem[(Q3)] How does the interaction between the filter and consumer change
  when the attacker \emph{strategically} crafts its attack in
  anticipation of this interaction? How does this affect the cost-benefit tradeoff for improving the filter quality?

To answer these
questions, we model
content filtering as a game between a filter and an information consumer.
The filter receives a
batch of content, wherein each piece is either legitimate or malicious with some exogenously specified probability.
For each piece of content,
the filter receives a signal regarding its
legitimacy, and either blocks
it or forwards it to the consumer.
In the latter case, the consumer exerts
costly effort to examine the content
and then decides whether to accept
or ignore it.
Both players benefit when the consumer accepts legitimate content,
and incur a cost when it does not consume legitimate content or consumes malicious content.
In an extension, an \emph{endogenous} attacker
sets the mean amount of the malicious content (\emph{attack propensity}) to maximize the expected amount of malicious content the consumer ultimately accepts.


\begin{prewebconf}
First the filter receives a batch of content
where each piece can be either legitimate or malicious.  For each piece of
content, the filter then receives a signal that is correlated with the
content's type and then strategically decides to either block the
content or release it to the consumer.  While real-world filters are
often autonomous systems operating in real time that follow fixed
rules, the filter's strategic behavior represents the filter's owner
optimally tuning the filter's parameters that set the sensitivity and
specificity.  Upon being exposed to content the consumer makes two
decisions.  First, it exerts costly effort to learn about whether the
content is legitimate or malicious.  It then decides to either accept
the content or ignore it.  The model with an endogenous attacker
proceeds similarly but instead of the content having a fixed
probability of being malicious, the attacker chooses an attack
frequency that then impacts the probability that any piece of content
is malicious prior to any examination.
\end{prewebconf}

The key novelty is that the consumer strategically chooses the fidelity of its signal and incurs the corresponding information cost.  This represents strategic information acquisitions where consumers optimally trade off the physical and cognitive costs of obtaining higher fidelity signals with the benefit associated with the higher fidelity information.  We adopt \emph{rational inattention} \cite{sims1},
a standard model for consumer's information cost.
Specifically, 
cost is proportional to the expected drop in entropy between the
consumer's prior and posterior.\footnote{Despite alternatives \cite{milgrom,vives,zhong,alt1,gen1,attnsurvey},
rational inattention is widely adopted as a standard model for information costs
\citep{martin, bertoli,ravid,cycles,travel,hiring,importers},
in the absence of further behavioral evidence or assumptions \cite{caplin2016}.}
The filter may internalize these costs, aiming to maximize consumers' welfare.%
\footnote{Maximizing users' welfare is a common modeling choice and a reasonable proxy for many online platforms that indirectly profit from user engagement, \eg via advertising.}
We also consider a variant in which the filter does \emph{not} internalize the information costs,
\eg when
it only cares about detection rates, which may be the case when platforms compete in performance benchmarks.
We call these variants, resp., \emph{aligned utilities} and \emph{semi-aligned utilities}.



With this model, we answer our questions as follows:

\fakeItem[(A1)] With an exogenous attacker and aligned utilities,
increasing filter quality
is Pareto-improving, but only weakly (Theorem \ref{thm:dvdp}).
There is a ``barrier to entry": equilibrium outcomes improve only when the filter is
accurate enough.

\fakeItem[(A2)] A new inefficiency arises when we switch to semi-aligned utilities.
Since the filter does not internalize the consumer's information cost, the filter is biased toward forwarding more content. It may not be credible for the filter to block any content, thus introducing a
  Pareto inefficiency 
  (Theorem \ref{thm:semi-inef-gen}).
  However, this inefficiency vanishes once the filter is
  sufficiently accurate
  (Theorem \ref{thm:misaligned-escape}), upon which further
  increases to filter quality are Pareto improving (Theorem
  \ref{thm:semi-VoT}).  

\fakeItem[(A3)] With a strategic attacker, there are two
  surprising consequences:  the consumer does not examine any
  content in any equilibrium (Theorem~\ref{thm:eattack-zero}), and
  improving the filter can  make both the filter and the consumer worse off (Theorem \ref{thm:eattack}).
The attacker raises its attack propensity, and this outweighs the direct benefit of a more accurate filter.

The main practical implication of our results
is that rote marginal improvements in filter quality are not
unambiguously beneficial. These improvements should either be large
enough, or be coupled with other interventions (such as training to
decrease information costs), to avoid a damaging reduction in consumer
attention.

\begin{prewebconf}
\jgcomment{Marked for deletion/dramatic reduction}
There are several practical implications of these results.  First,
filters should not expect to see continuous, linear and smooth returns
to investments in filter quality since consumers will only begin to
trust a filter after it has reached a certain threshold of quality.
Beyond such threshold, the consumers will optimally adjust their
information costs in response to changes in filter quality, leading to
nonlinear improvements in outcomes. Second, in the case of an exogenous
attacker, filter providers that are rewarded simply based on their
type I and type II error rates will not necessarily act in the best
interest of the consumer.  Instead, content filters must be
incentivized to not manipulate the consumer's beliefs and impose
higher than socially optimal information costs.
\end{prewebconf}

\begin{prewebconf}
\jgcomment{Marked for Deletion/Dramatic Reduction}
With an endogenous attacker, increasing the filter's quality may
reduce equilibrium payoffs for both the filter and the consumer. The
narrow implication is that in this regime, filters should not invest in
quality improvements as they have a negative return.  However, given
the limited scope of the model, such a policy recommendation is
extreme.  A more tempered recommendation is that as filter quality
increases, filter providers should initiate programs that ensure a
lower bound on consumer information gathering, regardless of attacker
behavior. One example of such a program would be warning systems in
email filters that require consumers to verify content.  A social
media platform may initiate a policy that requires consumers to
provide links to source material when sharing sensitive
information.  Of course, a sufficient reduction in information
costs can also help filters and consumers escape this regime, so
training programs and increased digital literacy are all possible
policy interventions.
\end{prewebconf}


\begin{prewebconf}
\jgcomment{All the stuff about alternatives to information cost models
  can be condensed into one paragraph with a ton of references.
  Dramatic reduction here.}

To model the cost the consumer incurs from gathering information about
whether content is legitimate or malicious, we adopt the
rational inattention framework\cite{sims1}.  Specifically, we model
information costs as being proportional to the mutual information
between the consumer's prior and posterior.  Or equivalently, the
(expected) reduction in uncertainty between the consumer's prior and
posterior, where uncertainty is quantified as the entropy of the
distribution.  Of course, the consumer may also incur costs simply by
consuming content.  For example, it takes time to read a post on
social media, regardless of its veracity.  We abstract away from those
costs and \emph{only} consider the costs the consumer incurs by
deliberating on the veracity of the content.
\end{prewebconf}

\begin{prewebconf}
From a behavioral perspective, the consumer chooses its ``information
strategy'' by choosing a signal structure.  For example, the consumer
can \emph{choose} a signal such that when content is legitimate, it
receives a signal indicating it is indeed legitimate 80\% of the time
and an (incorrect) signal that the content is malicious 20\% of the
time.  While choosing a signal distribution may seem unnatural, this
dynamic is equivalent to the consumer committing to gather information
about the content type until it reaches a specific degree of
certainty.  So for example, a consumer that just received an email
asking for bank account information will commit to knowing that the
email is legitimate with 98\% certainty before entering its
credentials and will ignore a message anytime it is less than 60\%
certain it is legitimate.  Without information costs, the
optimal consumer strategy would be to learn the content type with
total certainty. However, more certainty is costly so the consumer
optimizes by trading off its information costs with the costs of
interacting with the content.
\end{prewebconf}

\begin{prewebconf}
It is important to note that mutual information is not the only way to
model information costs but instead, it provides an appropriate
compromise between generality and parsimony.  One viable alternative
to the mutual information based costs would be to model the
physical choice to gather information.  This is what happens in prior
work on auctions \cite{milgrom} and oligopoly \cite{vives}.  For
example, in the oligopoly setting, a firm can pay to join a trade
association which then reveals the value of a relevant random
variable.  
The main drawback of modeling physical costs is its lack of generality
since the information gathering process and associated costs may vary
widely across domains.  For these reasons, we adopt what has been
called a black-box approach to information gathering \cite{attnsurvey}
and do not model the physical information gathering process but simply
assume that more certainty is more costly and assign a cost based on
the certainty of the consumer's posterior.

Assigning costs to certainty is a common approach, though the mutual
information cost function we use is not the only such
``posterior-based'' cost function.  There are also generalizations and
adaptations of the basic mutual information cost
structure \cite{zhong,alt1,gen1} that allows for features such as
state-dependent costs of information gathering.  However, the mutual
information formulation is akin to the ``Cobb-Douglas model of
attention, and moves to richer models are best based on behavioral
evidence \cite{caplin2016},'' and since we do not have \emph{a priori}
behavioral evidence, we adopt the mutual information cost function for
both its parsimony and analytical tractability.  Nevertheless, we
generalize some of our main results to differentiable and convex cost
functions in the appendix.

Technological innovation has made raw information widely available in
massive quantities and information consumers face the challenge of
where to spend their valuable attention.  It is then no surprise that
there has been a surge in the application of rational inattention to
canonical economic scenarios.  Examples include quality signaling with
inattentive consumers \cite{martin}, migration decisions
\cite{bertoli}, and bargaining \cite{ravid}, all of which use the
mutual information specification we use.  Other work on voting
decisions \cite{trombetta} is an example of rational inattention with
general convex cost functions similar to our extension in the
appendix.  Our model adds to this growing literature by showing how
introducing information costs to a traditional strategic decision
problem can lead to unforeseen behavior and social inefficiencies that
remain hidden when consumers do not endogenously acquire information.
\end{prewebconf}

Conceptually, we identify strategic interaction between content filters and information consumers as a relevant aspect of content filtering.
In contrast, prior game-theoretic work on content filtering studies games between filters and attackers \cite[e.g., ][]{kalgame,md}, between filters and a mediator \cite{recstrat}, or between consumers \cite{ace}.   Adversarial machine learning \cite{adbook1,adbook2} studies attacks on machine learning algorithms (such as content filters).
In all this work, consumers naively follow the filter's recommendations. We show that filter-consumer strategic interaction is
not captured by attacker-filter games.

While our model may appear similar to models in information design \cite{id1,cando}, and especially information design with rational inattention \cite{matyskova},
these models are fundamentally different:
senders can design arbitrary Blackwell \cite{bw} experiments that generate the receiver's signal.  In our model, the filter chooses an \emph{action} that has a direct impact on utility as well as consumer beliefs.  This coupling between actions and consumer beliefs is what sets our model apart from those of information design and yields new results. 

Our model is
similar to \cite[][P2020 for short]{papa} in that they both consider binary environments where a filter and consumers inspect content before choosing an action.  However, because consumers in our model choose their signal quality and the filter's signal is noisy (unlike that in P2020), we
examine the utility and behavioral impacts in changing filter quality, which is absent in P2020.   Additionally, we extend the environment and consider an endogenous attacker, another feature not included in P2020.






All proofs are deferred to Appendix \ref{sec:proofs}.

\section{Our model and preliminaries}
\label{sec:model}
\begin{prewebconf}
\jgcomment{Marked for deletion.  Don't need to outline the model,
  already kind of did in the intro. Just
  present it here.}
We consider the \emph{\ourGame}: a game between two strategic players,
an information filter and an information consumer. We call them the
\emph{filter} and the \emph{consumer}, and denote the respective
notation with subscripts $\fil$ and $\con$. The game focuses on a
single piece of content (henceforth, the \emph{content}) which can be
either legitimate or malicious. The content may correspond, \eg to an
email, a social media post, a recommendation, or a login prompt. The
content's legitimacy is set exogenously, and is only revealed to the
players via imprecise and noisy signals: via a classification
algorithm (for the filter) and costly manual examination (for the
consumer). Both players are incentivized to ensure that legitimate
content is accepted as such, while malicious content is
ignored. However, the consumer's \emph{information costs} for
examining the message lead to a complex strategic interaction which we
explore in the subsequent analysis.

\jgcomment{Marked for deletion/reduction.  Again, don't need to
  describe the model in words given space constraints}
The game unfolds in several discrete steps, whereby the filter
examines the content and decides whether to forward it to the
consumer, and then the consumer examines the content and decides
whether to accept it as legitimate or ignore it.  Mostly for
simplicity, we maintain binary assumptions.  That is, the filter only
receives one of two signals and can choose to either forward the
content or block it.  The binary signal assumption is for convenience.
For real world situations, filters likely receive a probabisitic
assessment (from a machine learning model, for example) and make a decision based
on that probability.  Our binary assumption captures the case where
the filter's actions are to either forward or block based on an
exogenously set threshold.  While future work could endogenize this
threshold, our binary model exhibits sufficient richness that we
abstract from the threshold-setting problem and simply assume binary
signals.  Furthermore, some real world filters can have multiple
actions such as forwarding content with a warning.  While expanding
the filter's action space is yet another viable extension to our
simple model, real world evidence suggests that warnings are often
ignored \cite{ignorewarn}, lending at least a partial justification
for our binary action assumption.
\end{prewebconf}



We consider the \emph{\ourGame}: a game between two strategic players,
an info filter and an info consumer that make
decisions about content's legitimacy. We call them the \emph{filter}
and the \emph{consumer}, and denote the resp. notation with
subscripts $\fil$ and $\con$.  The game's protocol is as follows:

\begin{prewebconf}
\begin{enumerate}
\item The filter receives a batch of content that contains $N_0$
  pieces of malicious content and $N_1$ pieces of clean content.  We
  write $q=\frac{N_0}{N_0+N_1}$ as the prior probability that a
  randomly selected piece of content is malicious, which is common
  knowledge.  We denote by $X$ a randomly drawn piece of content which
  is then malicious with probability $q$ and legitimate with
  probability $(1-q)$.  In section \ref{sec:endog}, we introduce a
  modified version of the game in which the attacker endogenously sets
  $N_0$.
\item The players treat each piece of content independently using $q$
  as the prior probability that the content is malicious.  For each
  piece of content:
\begin{enumerate}
\item The filter receives a private signal
  $\sfil\in\bin$ about the content type. The signal represents the
  output of a classifier, so that $\sfil=0$ means ``likely malicious"
  and $\sfil=1$ means ``likely legitimate". The signal is drawn
  independently from a known conditional distribution. We write
\begin{align}\label{eq:pi-defn}
    \pi_x = \Pr\sbr{\sfil=0\mid X=x} \qquad\text{for $x\in\bin$}
\end{align}
Loosely, $\pi_0$ is the filter's true positive rate where $\pi_1$ is
the false positive rate.   Without loss of generality, we assume that
$\pi_0 \leq \pi_1$ since the filter has the freedom to choose its
action conditional on its signal.
\item After receiving a signal, the filter chooses its \emph{action}
  $\afil\in\bin$: whether to block the content ($\afil=0$) or to
  forward it to the consumer ($\afil=1$).
\end{enumerate}
\item Upon receiving the unblocked batch of content, the consumer
  processes each piece of content independently.  For each piece of
  content:
\begin{enumerate}
  \item The consumer chooses how much
  effort to put into examining the content. Formally, the consumer
  chooses a conditional distribution over binary signals $\scon$ about
  the content type, which we refer to as an \emph{information
    strategy}.  The chosen distribution is characterized by
  $\mu=\rbr{\tilde{\pi}_0, \tilde{\pi}_1}$, where
  $\tilde{\pi}_x=\Pr\sbr{\scon=0\mid X=x}$.
\item The consumer then examines the content. Formally, the consumer
  receives a signal $\scon\in\bin$ about the content type, drawn
  independently from the chosen distribution.
    Here, $\scon=0$ means ``likely malicious" and $\scon=1$ means ``likely legitimate".
  \item The consumer chooses its \emph{action} $\acon\in\bin$: whether
    to accept the content as legitimate ($\acon=1$) or to ignore it
    ($\acon=0$).
\end{enumerate}
\end{enumerate}

Before proceeding there are two modeling choices to justify.  First,
we model a batch of content instead of a single piece of content to
highlight an important payoff scaling that has qualitative
implications for the case of the endogenous attacker, as will be shown
below.  Second, the fact that the filter and consumer analyze each
piece of content independently is a technical convenience to capture
the fact that the players don't know exactly how many of each type of
content there are but really only know $q$, the prior probability that
any piece of content is malicious.  This may happen if, for example,
the $N_0+N_1$ piece of content are constructed from a set of benign
content that arrive at a fixed Poisson rate and malicious content
that arrive according to another Poisson
process.  




Here is another example of what the protocol would look like if we use
Poisson rates:
\end{prewebconf}

\fakeItem[1.] The filter receives a batch of content (\eg a day's
worth of news).  The batch consists of malicious content that arrives
at a Poisson rate of $\rho_0$ and legitimate content that arrives at a
Poisson rate of $\rho_1$, per unit time interval.  Both rates
are common knowledge.  W.l.o.g., we normalize $\rho_1=1$.

Each piece of content in the batch is identified with a binary random variable
$X$, where $X=0$ means ``malicious" and $X=1$ means ``legitimate.'' We
define \[ q := \Pr[X=0] = \rho_0/(\rho_0+1).\]




\fakeItem[2.] Each piece of content $X\in\bin$ is processed by the {\bf filter} as follows.
The filter receives a private signal $\sfil\in\bin$ about the content type, representing the output of a classifier so that $\sfil=0$ means ``likely malicious" and $\sfil=1$ means ``likely legitimate". The signal is drawn independently from a known conditional distribution given $X$. Denote the resp. true and false positive rates as
\begin{align}\label{eq:pi-defn}
    \pi_x = \Pr\sbr{\sfil=0\mid X=x}, \quad x\in\bin.
\end{align}
W.l.o.g. assume
$\pi_0 \geq \pi_1$ (since the filter is free to choose its
action conditional on its signal).
After receiving the signal, the filter chooses its \emph{action}
  $\afil\in\bin$: whether to block the content ($\afil=0$) or to
  forward it to the consumer ($\afil=1$).

\fakeItem[3.] Each piece of forwarded content is processed by the {\bf consumer} as follows.
The consumer chooses how to examine the content. Formally, the consumer controls the distribution of a signal $\scon\in\bin$, where $\scon=0$ means ``likely malicious" and $\scon=1$ means ``likely legitimate".
The signal is drawn independently from some conditional distribution given $X$, characterized by
\begin{align}\label{def:con-probs-defn}
 \tilde{\pi}_x=\Pr\sbr{\scon=0\mid X=x}, \quad x\in\bin.
\end{align}
These probabilities are chosen by the consumer in advance, at the (information) cost specified below.
Then, the consumer chooses its \emph{action} $\acon\in\bin$: whether
    to accept the content as legitimate ($\acon=1$) or to ignore it
    ($\acon=0$).


    \xhdr{Strategies.}  The filter and the consumer have pure action strategies
    $\strfil, \strcon: \bin\to\bin$ so that $\afil = \strfil(\sfil)$
    and $\acon = \strcon(\scon)$.  The consumer also chooses
    probabilities $\mu=\rbr{\tilde{\pi}_0, \tilde{\pi}_1}$ from
    \refeq{def:con-probs-defn}, called its \emph{information
      strategy}. Thus, pure strategies are $\strfil$ for the filter,
    and $(\strcon,\mu)$ for the consumer. Both players choose their
    (mixed) strategies before the game starts, and those strategies are applied to
    the entire batch. (This is justified because the pieces of content
    are ex-ante equivalent.) We posit that the filter and the consumer
    choose their (mixed) strategies simultaneously, \ie without
    observing one another.


\begin{rem}\label{rem:model-order}
  When the filter and consumer have fully aligned utilities
    (as defined below and discussed in Sections~\ref{sec:aligned},
    \ref{sec:endog}), our results carry over to the variant
    where the players choose their mixed strategies sequentially: the
    filter moves first, and the consumer best-responds. This is
    because our results focus on the socially optimal strategy
    profile (defined in Section \ref{sec:aligned}), which is the same in both variants.
\end{rem}

\begin{rem}\label{rem:model-zero}
One pure strategy for the consumer is to \emph{not} examine the content and  incur no info cost.
\end{rem}




\xhdr{Notation.}
A generic mixed strategy profile is denoted $\mixsigma$. The players' mixed action strategies are, resp., $\mixstrfil$ and $\mixstrcon$.

We label three filter pure strategies:
the \emph{blocking strategy} $\strblk$
which always blocks the content:
    $\strblk(\cdot) \equiv 0$,
the \emph{forwarding strategy} $\strfwd$
which always forwards the content:
    $\strfwd(\cdot) \equiv 1$,
and the \emph{differentiating strategy} $\strdif$
which differentiates between the signals:
    $\strdif(\psi)\equiv \psi$.
We ignore the ``unreasonable  strategy''  in which the filter forwards
``likely malicious'' content and blocks content that is ``likely clean'' as it can never be part of a non-trivial equilibrium (see \UnreasonableStrategyLocation for technical details).  

A strategy profile is called \emph{consumer-optimal} if the consumer best-responds to the filter's strategy.
Let the \emph{blocking profile} $\blkprof$,
the \emph{forwarding profile} $\fwdprof$,
and the \emph{differentiating profile} $\difprof$,
be consumer-optimal strategy profiles in which
the filter's pure strategy is, resp.,
$\strblk$, $\strfwd$, and $\strdif$.

\begin{center}
{\small
\begin{tabular}{r|l|l}
    & $\strfil(1)=0$
    & $\strfil(1)=1$
\hlinestrut
$\strfil(0)=0$
    & \emph{blocking profile} $\blkprof$
    & \emph{differentiating profile} $\difprof$
\hlinestrut
$\strfil(0)=1$
    & ($\strfil$ is ``unreasonable'')
    & \emph{forwarding profile} $\fwdprof$
\end{tabular}}
\end{center}

\xhdr{Utilities.}  The consumer's utility per piece of content is the
difference between the \emph{action payoff} $u(\afil\cdot\acon,X)$,
determined by how the actions match the content type, and the
\emph{information cost} for examining the content.
We interpret the product $\afil\cdot\acon\in\bin$
as an aggregate action: indeed, the content is accepted
if $\afil\cdot\acon=1$, and ignored otherwise.
The consumer receives a reward when
legitimate content is accepted ($\afil\cdot\acon=X=1$), and penalties
if the content is misclassified ($\afil\cdot\acon\neq X$). We
normalize action payoffs to $0$ if malicious content is ignored
($\afil\cdot\acon=X=0$).
Thus, action payoffs $u(\afil\cdot\acon,X)$ are summarized by a $2\times 2$ table below, with $b,c_1,c_2\geq 0$.
\begin{center}
\begin{tabular}{l|c|c}
   & $X=0$  & $X=1$  \\ \hline
$a_{\fil}\cdot a_{\con}=0$  & 0      & $-c_1$ \\ \hline
$a_{\fil} \cdot a_{\con}=1$ & $-c_2$ & $b$    \\
\end{tabular}
\end{center}

The information cost is the cost of obtaining signal $\scon$ about content type $X$.  
It is proportional to how far the consumer's beliefs shift away from its prior, and only accrues when the filter does not block content.
More abstractly, we define the information cost for obtaining some randomized signal $\signal$ about some hidden state $X$ given some event $\mE$, denoted
    $\infoC\sbr{\signal; X \mid \mE}$
and determined by the conditional joint distribution of $(\signal,X)$
given $\mE$.
We adopt the (widely accepted) definition from \citet{sims1}.
\begin{align}\label{eq:cost-defn}
 \infoC\sbr{\signal;X \mid \mE}
    = \infoL\cdot I\rbr{\signal;X \mid \mE},
\end{align}
where $I\rbr{\signal;X \mid \mE}\geq 0$ is the mutual information conditional on the event $\mE$ and $\infoL>0$ is a known parameter. Thus, the information cost for examining the content is defined via \eqref{eq:cost-defn} as
    $\infoC\sbr{\scon;X \mid \afil=1}$.
Note that the cost indirectly depends on filter's mixed action strategy since information costs are a function of the consumers prior upon receiving content, which depends on the filter's strategy.



\begin{prewebconf}
Since
    $I(\signal;X) = H(X)-H(\signal|X)$,
the information cost \eqref{eq:cost-defn} is proportional to the
expected reduction in consumer's uncertainty about $X$, where the uncertainty is expressed via entropy $H$. The information cost is $0$ if and only if the signal is \emph{uninformative}, in the sense that
    $\Pr\sbr{X=\cdot\mid \signal,\mE}$
does not depend on $\signal$. The information cost
    $\infoC\sbr{\scon;X \mid \afil=1}$
depends not only on the conditional distribution of $\scon$ chosen by the consumer, but also on the filter's action strategy
$\mixstrfil$, because $\mixstrfil$ impacts the posterior distribution
given $\afil=1$.
\end{prewebconf}

The consumer's expected payoff per a random piece of content $X$ under mixed strategy profile $\mixsigma$ is therefore
\begin{align*}
  \vcon(\mixsigma) = \E\sbr{ u(\acon\cdot\afil,X) - \afil\cdot \infoC\sbr{\scon;X \mid \afil=1}},
\end{align*}
where the expectation is over $X,\sfil,\scon,\mixsigma$.
As a shorthand, let
    $u(\mixsigma) = \E\sbr{ u(\acon\cdot\afil,X)}$
and
   $\infoC(\mixsigma) = \infoC\sbr{\scon;X \mid \afil=1}$
be the corresponding expected action payoff and information cost.

The consumer's total expected utility over the batch is
\begin{align}
  \Vcon(\mixsigma)
    =(1+\rho_0)\;\vcon(\mixsigma)
    = \vcon(\mixsigma) / (1-q).
\end{align}
where $1+\rho_0$ represents the expected batch size.


To define the filter's utility, we consider two variants.
The main variant (\emph{aligned utilities}) is that the filter's utility equals the consumer's. We also consider another variant (\emph{semi-aligned utilities}) when the filter internalizes the action costs but not the information costs.
Let $\Vfil(\mixsigma)$ be filter's total expected utility under profile $\mixsigma$.
Then
  $\Vfil(\mixsigma) = \Vcon(\mixsigma)$ for aligned utilities, and
  $\Vfil(\mixsigma) = u(\mixsigma)/(1-q)$ for semi-aligned utilities.

\begin{prewebconf}
 \xhdr{Solution concept.}  Our equilibrium concept is perfect
  Bayesian equilibria (PBE).  A strategy profile $\mixsigma$ together
  with a beliefs $q_{\fil}, q_\con$  about the content type $X$ is
  a PBE of our game if: a) player $i$'s strategy is
  utility-maximizing given beliefs $q_i$, b) player $i$'s belief
  $q_i$ is consistent with a Bayesian update given the prior
  distribution $q$ and observed signals.

  Of course in equilibrium, the consumer and filter will always be
  best responding.  However, throughout the analysis it is helpful to
  fix the filter's strategy and consider the profile in which the
  consumer best responds to the filter's strategy \emph{without}
  supposing optimality of the filter.  We call such a strategy profile
  a \emph{consumer-optimal strategy profile}.
\end{prewebconf}

\begin{prewebconf}
\jgcomment{Mark to cut}
\xhdr{Model parameters.} Another useful distinction we will reference
is that  our model is specified by two different sets of  parameters: The
\emph{payoff parameters} $b,c_1,c_2,\infoL$ which characterize action
payoffs and information costs, and \emph{distribution parameters}
$q$, $\pi_0$,  $\pi_1$ which characterize the
distribution of $X$ and the joint distribution of
$(X,\sfil)$, respectively.
\end{prewebconf}

\xhdr{Value of Technological Change.}  We are particularly interested
in how improving the technology impacts equilibrium outcomes.
Specifically, we consider improving the quality of the filter, in
terms of raising $\pi_0$ and/or lowering $\pi_1$.\footnote{Filter's quality takes two numbers to describe.}
We adopt Perfect Bayesian Equilibrium (PBE)
as a solution concept \cite{mwg}.


For concreteness, fix some equilibrium selection rule, $f$, \cite{esel} and filter quality parameters, $\pi_0$ and $\pi_1$.  For each
player $i \in \cbr{\fil,\con}$, let $V^f_i(\pi_0,\pi_1)$ be $i$'s
equilibrium payoff under this rule.  We are interested in
the difference in equilibrium payoffs between a high- and
  low-quality filter:
\begin{align}\label{eq:vot-def}
V^f_i(\pi'_0,\pi'_1) - V^f_i(\pi_0,\pi_1):
    \quad i \in \cbr{\fil,\con},
\end{align}
where $\pi'_0\geq\pi_0$ and $\pi'_1\leq\pi_1$. We call \eqref{eq:vot-def}
the \emph{value of technological change} (\vot).
We say that \vot is positive (resp., negative) if \refeq{eq:vot-def} is that way for both players, \ie if improving the filter Pareto-increases (resp., Pareto-decreases) equilibrium payoffs.

Consider \vot under infinitesimal filter improvement:
\begin{align}\label{eq:mvot-def}
 \frac{\partial}{\partial \pi_0} V_i^f(\pi_0,\pi_1),\;
     -\frac{\partial}{\partial \pi_1} V_i^f(\pi_0,\pi_1):
    \; i \in \cbr{\fil,\con},
\end{align}
assuming the partial derivatives in \eqref{eq:mvot-def} are well-defined.
We call \eqref{eq:mvot-def} the \emph{Marginal Value of Technology} (\mvot).
The \mvot specifies how much a rational filter would pay to improve its quality. A zero (resp., negative) \mvot means the filter would not pay anything (resp., would have to \emph{be paid}).

\section{Consumer Beliefs}
\label{sec:trivial}

This section  presents a preliminary analysis of consumer behavior,
which applies to both aligned and semi-aligned utilities,
and serves as scaffolding for what follows.

An important quantity is the consumer's belief that the forwarded
content is malicious,
given that the filter's mixed action strategy is $\mixstrfil$.  We define
this quantity as:
\begin{align}\label{eq:belief-defn}
\qfunc(\mixstrfil) := \Pr\sbr{X=0\mid \mixstrfil,\,\afil = 1}.
\end{align}

\noindent Note that $\qfunc(\strfwd)$ is simply $q:=\Pr[X=0]$.

The following lemma shows that the consumer's behavior is uniquely
determined by $q(\mixstrfil)$:

\begin{lemma}\label{lem:belief-BR}
Given any filter  mixed strategy $\mixstrfil\neq \strblk$
the consumer's best response to $\mixstrfil$ is determined by $\qfunc(\mixstrfil)$.
\end{lemma}




As per Remark~\ref{rem:model-zero},
the consumer can choose to \emph{not} examine the content
and incur no information costs.
Below we establish a regime where that is indeed optimal.  Define:
\begin{align}
1 &>\qH := \frac{\exp(b/\infoL) -\exp(-c_1/\infoL)}
        {\exp(b/\infoL)-\exp(-(c_1+c_2)/\infoL)} \nonumber \\
  &>\qL := \qH\cdot \exp(-c_2/\infoL) >0. \label{eq:HL-defn}
\end{align}
\begin{proposition}\label{prop:goalposts}
Let $\mixsigma$ be a consumer-optimal mixed strategy profile with
filter's mixed action strategy $\mixstrfil \neq \strblk$. Then  $\infoC(\mixsigma)=0$ if and only
if $\qfunc(\mixstrfil)\not\in (\qL,\qH)$.  Furthermore, if
$\qfunc(\mixstrfil) \leq \qL$ the consumer's optimal strategy is to
accept all content.  If $\qfunc(\mixstrfil) > \qH$ the consumer's
optimal strategy is to ignore all content.
\end{proposition}


In words, if unblocked content is too likely to be malicious (resp.,
legitimate) for a given $\mixstrfil$, the consumer's
best-response is to ignore (resp., accept) it without examination.


\begin{rem}\label{rem:params}
The quantities $\qfunc(\mixstrfil)$, $\qH$, $\qL$
are meaningful as per Proposition~\ref{prop:goalposts}.
They usefully encapsulate the numerous parameters in our model, and
are essential in our subsequent results.
Note that
\eqref{eq:belief-defn} is determined by the joint distribution of $X$
and the filter's signal $\sfil$, whereas \eqref{eq:HL-defn} is
determined by all parameters related to the costs.
\end{rem}

We now derive the \mvot under some consumer-optimal profiles in
some parameter regimes. A key quantity here is
\begin{align}\label{eq:eff-defn}
\qdif := \qfunc(\strdif)
      =\Pr\sbr{X=0\mid \strdif, \sfil = 1}<q,
\end{align}
where the inequality follows because $\pi_0\geq\pi_1$.

\begin{proposition}
  \label{cor:scaff-VoT}
For $i\in\cbr{\fil,\con}$ and $x\in\bin$:
\begin{itemize}

\item[(a)] \emthm{Zero  \mvot.}
$ \frac{\partial}{\partial\pi_x} V_i(\fwdprof)= \frac{\partial}{\partial\pi_x} V_i(\blkprof)= 0$.\\
If $\qdif>\qH$ then
    $ \frac{\partial}{\partial\pi_x} V_i(\difprof)= 0$.

\item[(b)] \emthm{Constant \mvot.}
If $\qdif<\qL$ then
\begin{align*}
\frac{\partial}{\partial\pi_0} V_i(\difprof)&= \frac{q}{1-q}c_2 >0 \\
\frac{\partial}{\partial\pi_1} V_i(\difprof) &= -(c_1+b)<0.
\end{align*}


\end{itemize}
\end{proposition}

In words, there is no benefit to improving the filter
if the filter's action does not depend on its signal, or
the consumer's best response is simply to ignore all content.
On the other hand, if the consumer  accepts the filter's
recommendation, then \mvot is constant.
To fully characterize the \mvot, 
subsequent analysis will
focus on deriving the \mvot when $\qL<q(\mixsigma)<\qH$ \emph{and}
establishing which profile constitutes an equilibrium.

\section{Aligned utilities ($\Vfil= \Vcon$)}
\label{sec:aligned}

In this section, we consider \emph{aligned utilities}.  Let
$V:= \Vfil = \Vcon$.  We focus on socially optimal profiles (ones that
maximize $V$), noting that any such profile is an equilibrium. Let
$\Vopt = V^f_i(\pi_0,\pi_1)$, where $f$ chooses the equilibrium that maximizes $V$ among all equilibria.  Our first
result is that $\Vopt$ has a simple characterization in terms of two
pure profiles defined in Section~\ref{sec:model},
the differentiating profile $\difprof$ and forwarding profile $\fwdprof$:%

\begin{proposition}\label{prop:aligned-best-eq}
$\Vopt = \max\rbr{V(\difprof),\, V(\fwdprof)}$.
\end{proposition}

While it is straightforward to algebraically demonstrate which of these two profiles are the best among the \emph{pure} strategy profiles, it is more difficult to prove that
there is no benefit from the filter using a mixed strategy.
Indeed, in our game the payoffs at a mixed equilibrium are not (necessarily) linear in the mixing probabilities, because the latter enter non-linearly in the information costs. Consequently, it is no longer trivially guaranteed that some pure strategy profile is socially optimal.



The main result here 
fully characterizes the marginal value of technological change (\mvot)
in terms of $\Vopt$.

\begin{theorem}\label{thm:dvdp}
  $\;$
\begin{itemize}

\item[(a)] \emthm{Zero \mvot.}
Suppose
      $\qdif>\qH$ or $V(\difprof)<V(\fwdprof)$.\\
Then
$\partial \Vopt / \partial \pi_0 = \partial \Vopt / \partial \pi_1 = 0$.

\item[(b)] \emthm{Constant \mvot.}
If $\qdif<\qL$ and $V(\difprof)>V(\fwdprof)$, 
\begin{align*}
\frac{\partial{\Vopt}}{\partial \pi_0} = \frac{q}{1-q}c_2>0
\quad\text{and}\quad 
\frac{\partial{\Vopt}}{\partial \pi_1} = -(c_1+b)<0.
\end{align*}

\item[(c)] \emthm{Non-constant \mvot.}
Suppose $\qdif\in (\qL,\qH)$ and $V(\difprof)>V(\fwdprof)$. Then
\begin{align}\label{eq:thm:dvdp-nonlinear}
\frac{\partial{\Vopt}}{\partial \pi_0}
    = \frac{q}{1-q}\infoL\cdot\log\rbr{\frac{\qH}{\qdif}} > 0
\\ \nonumber
\frac{\partial{\Vopt}}{\partial \pi_1}
    = \infoL\cdot\log\rbr{\frac{1-\qH}{1-\qdif}}< 0.
\end{align}
\end{itemize}
\end{theorem}




The main insight of Theorem~\ref{thm:dvdp} is that
\mvot is weakly \emph{but not strictly} positive. 
That is, when
incentives are aligned, improving the filter quality can never hurt
the players, though in some cases it may have no impact. 
Moreover, we fully characterize \mvot behavior based on how $\qdif$ compares with $(\qL,\qH)$, and whether
$V(\difprof)<V(\fwdprof)$.
This is summarized in the table below.
(In this table, \vot is positive in both cells in which it is not zero.)

\vspace{2mm}
{\small
\begin{center}
\begin{tabular}{r|l|l|l}
    & $\qdif <\qL$
    & $\qdif\in(\qL,\qH)$
    & $\qdif>\qH$
\hlinestrut
$\Vopt(\difprof)>$ \\ $\Vopt(\fwdprof)$
    & Constant
    & Non-linear
    & Zero
\hlinestrut
$\Vopt(\difprof)< $\\ $\Vopt(\fwdprof)$
    & \multicolumn{3}{c}{Zero \vot}
\end{tabular}
\end{center}}
\vspace{2mm}


We have two \emph{barriers to entry} in filter
technology.  First, recall that we have zero \mvot when $\qdif>\qH$,
and note that the filter quality is higher for lower values of
$\qdif$.  Therefore, the filter must be of sufficiently high quality
for improvements to make a difference.

Second, if $V(\fwdprof)>V(\difprof)$ then improving the filter does
not help, either.  In particular, the forwarding profile $\fwdprof$ is
now socially optimal, and so the filter is better off forwarding all
content regardless of its signal.  The next proposition shows that there exists parameter regimes where the socially optimal equilibrium is 
one in which the \mvot is $0$. To this end, we characterize this
regime precisely in terms of the model fundamentals.

\begin{proposition}\label{prop:diff}
Let $\KL{p}{q}$ be the Kullback-Leibler divergence between Bernoulli distributions with success probabilities $p$ and $q$.  Then
$V(\difprof)  \geq V(\fwdprof)$
if and only if one of the following conditions hold:

\begin{itemize}
\item[(a)] $q \geq \qH$,

\item[(b)] $\qL<\qdif<q<\qH$ and
\[ \deltaU >\infoL\sbr{\KL{q}{\qL} - \beta\cdot \KL{\qdif}{\qL}},\]

\item[(c)]
   $\qdif\leq \qL<q<\qH$
and
    $\deltaU > \infoL\cdot\KL{q}{\qL}$,

\item[(d)]
$\qdif<q\leq \qL<\qH$ and $\deltaU>0$,
\end{itemize}


\noindent where we used the following shorthand
\[ \deltaU := \E\sbr{u\rbr{\strdif(\sfil),X}} - \E\sbr{u\rbr{\strfwd(\sfil),X}}, \]
for the expected increase in action payoffs with an always-accepting consumer when the filter's  strategy switches from $\strfwd$ to $\strdif$;
and $\beta := \Pr\sbr{ \afil=1 \mid \strfil = \strdif}$
is the \emph{ex ante} probability that a differentiating filter forwards the content.%
\footnote{Unwrapping,
$ \deltaU = \pi_0qc_2 -\pi_1(1-q)(b+c_1)$
and
$\beta =(1-\pi_0)q + (1-\pi_1)(1-q)$.}
\end{proposition}

It is straighforward to show that
both barriers are cleared once the filter quality is high enough:


\begin{corollary}
  \label{prop:aligned-overcome}
There exist thresholds $\pi'_0<1$ and $\pi'_1>0$ such that
     $\qdif<\qH$
and
     $V(\difprof) > V(\fwdprof)$
for any
    $\pi_0>\pi'_0$ and $\pi_1<\pi'_1$.
\end{corollary}



Finally, Proposition~\ref{prop:diff} implies that the non-linear \vot regime from Theorem~\ref{thm:dvdp} is feasible. Indeed, this regime corresponds to case (b) of the proposition.

\medskip

\medskip

\begin{prewebconf}
\begin{figure}[t]
  \centering
\begin{subfigure}{.49\textwidth}
  \centering
  \includegraphics[width=\textwidth]{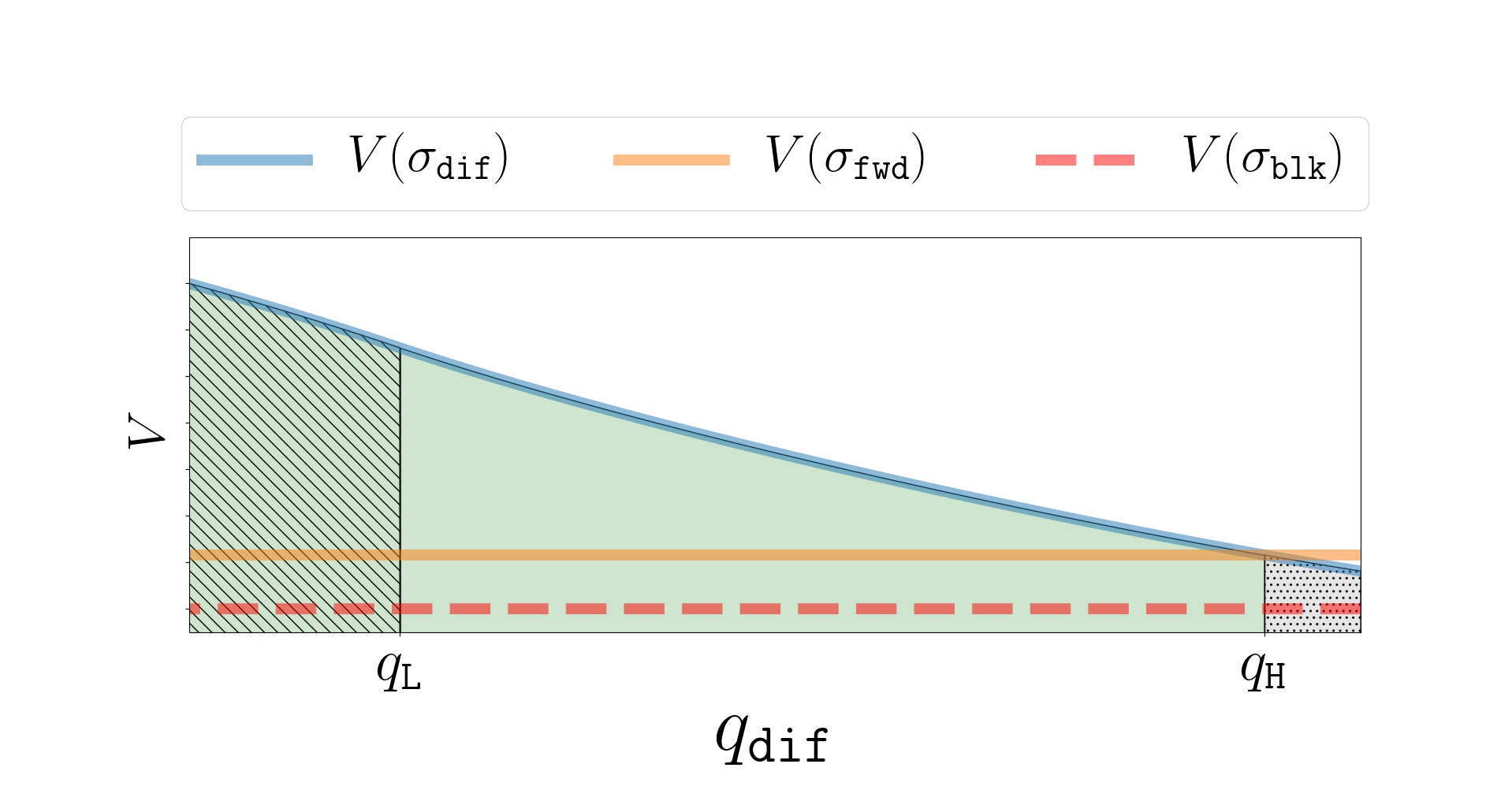}
  \caption{$q=.5$}
  \label{fig:noqh}
\end{subfigure}
\begin{subfigure}{.49\textwidth}
  \centering
  \includegraphics[width=\textwidth]{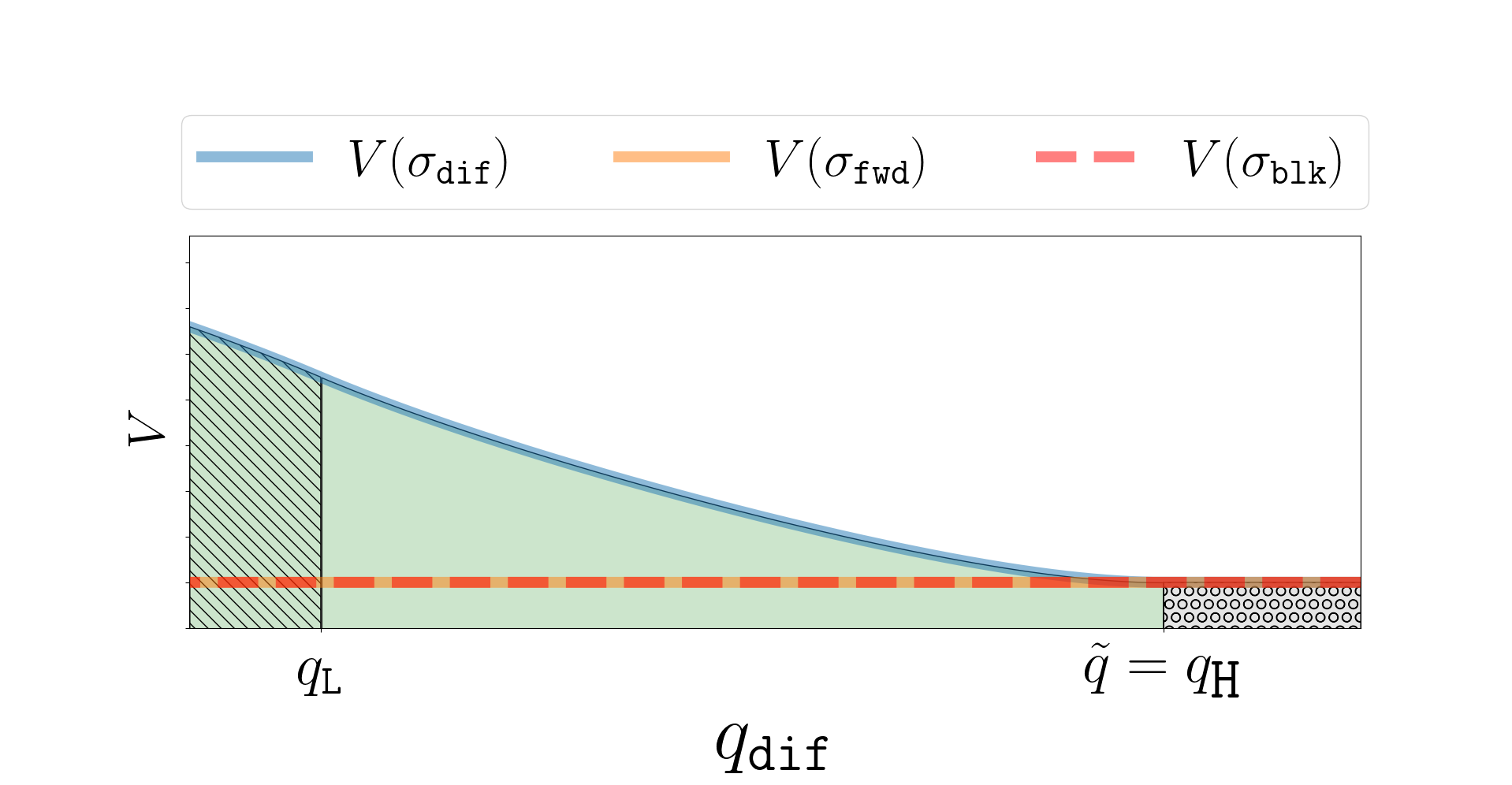}
  \caption{$q=.8$}
  \label{fig:qh}
\end{subfigure}
\caption{Changes in payoffs as a function of filter quality.  In both
  figures $\lambda=2, b=1, c_1=1,c_2=4, \pi_2=.3$.  In \ref{fig:noqh},
  $q=.5$ and in \ref{fig:qh}, $q=.8$.  The figures are generated by
  varying $\pi_0$.  The value $\tilde{q}$ is the point where the
  players are indifferent between $V(\fwdprof)$ and $V(\difprof)$.  }
\label{fig:alignmain}
\end{figure}
\end{prewebconf}

\section{Semi-aligned utilities ($\Vfil=u$)}
\label{sec:semi-aligned}




This section considers \emph{semi-aligned utilities}:
$\Vfil(\sigma) = u(\sigma)$. Our results concern Pareto-efficiency.
We show that all equilibria may be Pareto-inefficient
(in stark contrast with the aligned utilities),
but this inefficiency vanishes if the filter quality is sufficiently high.
Put differently, improving the filter has an important side benefit
of guaranteeing Pareto-efficient equilibria.

For clarity, we focus on the regime where $\qL<\qdif<q_H<q$  
(Similar
results holds for other regimes, but have a higher notation burden).
In this regime, the inefficiency arises when one measure of filter
quality is sufficiently low.  Specifically, we summarize filter
quality as one number that is strictly pointwise-increasing in $\pi_0$ and
$-\pi_1$,
\begin{align}\label{eq:Q-defn}
  \mathcal{Q}(\pi_0, \pi_1)
    = \frac{\pi_0}{\pi_1} \;
      \frac{1-\pi_1}{1-\pi_0}.
\end{align}
We compare \eqref{eq:Q-defn} to a threshold driven by cost parameters:
\begin{align}\label{eq:lambda-defn}
  \Lambda = \frac{b+c_1}{c_2}\frac{1-\qL}{\qL}.
\end{align}




\begin{theorem}[inefficiency]\label{thm:semi-inef-gen}
Assume $\qL<\qdif<\qH<q$.  If furthermore
    $\mathcal{Q}(\pi_0,\pi_1)<\Lambda$,
then profile $\difprof$ strictly Pareto-dominates any equilibrium but
is not  an equilibrium itself.
In particular, any equilibrium is Pareto-inefficient.
\end{theorem}

The key insight behind Theorem \ref{thm:semi-inef-gen} is that a low
quality filter cannot commit to $\difprof$ because it has an incentive
to trick the consumer into incurring information costs that are higher
than optimal for the consumer.
Under $\difprof$, the filter incurs a cost of $(1-q)\pi_1 c_1$ for
blocking clean content but incurs some benefit from the consumer's
content inspection.  If the filter could \emph{convince} the consumer
it would choose $\strdif$, the filter would be better off by instead
forwarding all content, not incurring the cost of $(1-q)\pi_1c_1$ and
still enjoying the benefit of the consumer inspecting the content.
Knowing this, the filter can \emph{not} convince the consumer that it
would play $\strdif$ and thus $\difprof$ is not an equilibrium.



To escape this inefficiency, one can improve the filter, ensuring that $\mathcal{Q}>\Lambda$.
The \vot would be strictly positive.

\begin{theorem}[escaping the inefficiency]
\label{thm:misaligned-escape}
Assume $\qH<q$ and suppose $\pi'_0\leq \pi_0$, $\pi_1'>\pi_1$ and
    $\mathcal{Q}(\pi_0, \pi_1)>\Lambda>\mathcal{Q}(\pi'_0, \pi'_1)$.
Then:
\begin{itemize}
\item[(a)] The differentiating profile $\difprof$ is a Pareto-efficient equilibrium,
and it Pareto-dominates any other equilibrium.

\item[(b)]
The \vot by switching  from any equilibrium with filter
quality $(\pi'_0,\pi'_1)$ to $\difprof$ with filter quality
$(\pi_0,\pi_1)$ is strictly positive.
\end{itemize}
\end{theorem}

The intuition behind Theorem~\ref{thm:misaligned-escape} is as
follows.  As soon as the filter is of sufficiently high quality,
$\difprof$ becomes an equilibrium, is Pareto efficient and furthermore,
is the equilibrium preferred by both players.  Behaviorally, when the
filter is of sufficiently high quality, it is credible for the filter to use strategy $\strdif$. Content with a strong bad signal is so likely to be malicious that the filter prefers not to forward it.  As a result, the
filter can credibly commit to playing $\difprof$.  This
characterization again highlights the non-linear nature of filter
improvements and the importance of the filter meeting a baseline level
of quality.  However, unlike in the aligned section where it was the
\emph{consumer} that would abandon platforms with low quality filters,
with semi-aligned incentives it is the \emph{filter}'s incentive to
forward too much content that leads to inefficient outcomes with low quality
filters.


Consider the regime of Theorem \ref{thm:misaligned-escape}(a),
\ie $\qH<q$ and $\mathcal{Q}(\pi_0, \pi_1)>\Lambda$.
Once the filter and consumer enter this regime, further improving the
filter would  keep them in that regime.
Theorem \ref{thm:semi-VoT} shows that such improvements would benefit both
players, and characterizes the resulting \vot.  

\begin{theorem}\label{thm:semi-VoT}
  Assume $\qL<\qdif<\qH<\qH$ and $\mathcal{Q}(\pi_0, \pi_1)>\Lambda$.
  Under equilibrium $\difprof$,   the \mvot is positive for both
  players, constant for the filter, and non-constant for the consumer.
\end{theorem}





\begin{prewebconf}

While it would be desirable to characterize \vot in the semi-aligned
case, the mismatch between the consumer and the filter leads to
multiple equilibria that are, in general, \emph{not} Pareto
comparable.  Nevertheless, we a) show there are only $3$ possible
Pareto-incomparable equilibrium profiles that are not weakly Pareto
dominated and b) derive the \vot under those profiles.  In other words, we
do not claim to resolve the equilibrium selection issue (which cannot
be resolved through Pareto dominance as in the case of aligned incentives), but we
derive the \vot under each joint profile that may be a Pareto
incomparable equilibrium.

\jgcomment{ I think we can delete the rest of this section}
Figure \ref{fig:semi-ineff} illustrates the intuition behind theorem
\ref{thm:semi-inef-gen}.  In the plot, $\tilde{q}$ is the value of
$\qdif$ such that the filter is indifferent between $\fwdprof$ and
$\difprof$.  The pink area between $\tilde{q}$ and $q_{\Lambda}$
represents the region where both players (the consumer's payoffs are
not shown) prefer $\difprof$ over $\fwdprof$. However, under the
consumer's best response to $\difprof$ the filter can do better by
instead choosing $\strfwd$ and thus $\difprof$ is not an equilibrium.
This is represented in the figure as $\Vfil(\strfwd, BR(\strdif))$
being greater than $\Vfil(\difprof)$.  However, once the filter
quality increases such that $\qdif>q_{\Lambda}$, the filter no longer
has an incentive to deceive the consumer and thus the inefficiency
vanishes in the green region.  Furthrmore, when $\qdif>q_{\Lambda}$,
$V(\fwdprof)$ is Pareto dominated by $V(\difprof)$ so if we assume
equilibrium selection chooses among non-Pareto dominated equilibria,
when $\qdif$ rises abouve $q_{\Lambda}$, there is a positive jump in
equilibrium payoffs as the players switch from $\fwdprof$ to
$\blkprof$.\footnote{Note the plot was generated such that $q<\qH$
  providing evidence to our assertion in the previous footnote that
  $q>\qH$ is not a necessary condition for the inefficiency to exist.}

\begin{figure}
  \centering
  \includegraphics[width=.8\textwidth]{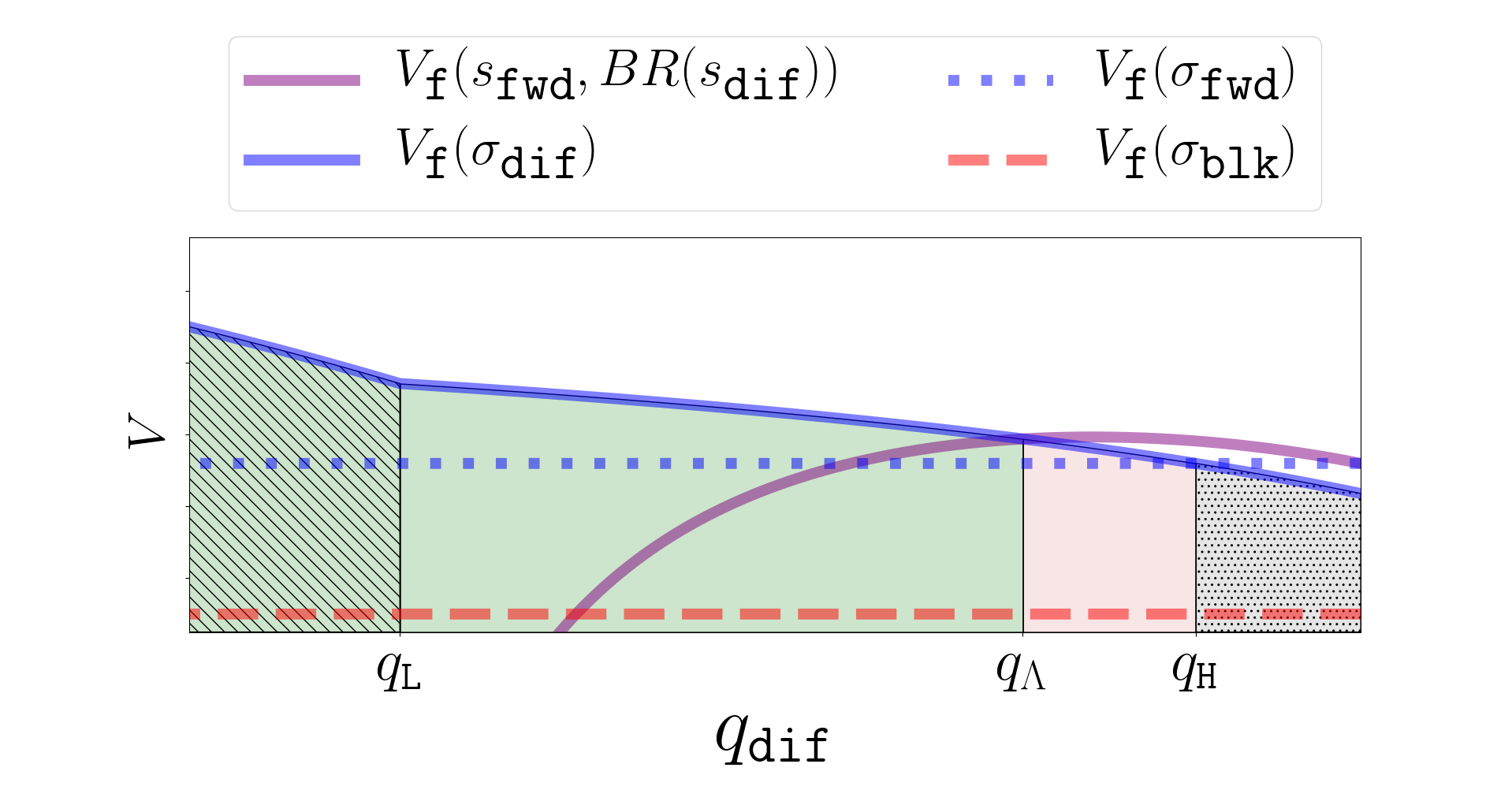}
  \caption{Filter's payoff under different strategy profiles.
    $\Vfil(\strfwd, BR(\strdif))$ is the filter's expected payoff when
    it chooses $\strfwd$ but the consumer best responds to $\strdif$.
    The parameters that generated this figure are $\lambda=2, b=1,
    c_1=1, c_2=4, q=.5, \pi_2=.2$ and $\pi_1$ goes from $.2$ to $1$.}
  \label{fig:semi-ineff}
\end{figure}

As previously mentioned, in the case of semi-aligned utilities,
equilibrium selection issues prohibits us from fully characterizing
\vot.  We now formalize this equilibrium selection problem and derive
the VoT under those equilibrium profiles.  We first establish that
there is only one properly mixed filter's strategy that can be an
equilibrium with higher payoffs than $V_i(\blkprof)$.  We then
identify in the same parameter regime the existence of $2$ pure
strategy equilibria.  Together, these results imply that there are at
most $3$ non-Pareto dominated equilibria.

\begin{proposition}\label{prop:semi-regimes}
  Assume $\qL<q<\qdif<\qH$ and
$\Lambda \in \rbr{\frac{\pi_0}{\pi_1},\, \frac{\pi_0(1-\pi_1)}{\pi_1(1-\pi_0)}}$
then both $\difprof$ and $\fwdprof$ are equilibria with payoffs higher
that $V_i(\blkprof)$.
\end{proposition}

\begin{proposition}\label{prop:semi-mixed}
  There is an equilibrium $\mixprof$ with (properly) mixed filter's
  strategy that strictly Pareto-improves over $\blkprof$ if and only if
  $\Lambda \in \rbr{\frac{\pi_0}{\pi_1},\,
    \frac{\pi_0(1-\pi_1)}{\pi_1(1-\pi_0)}}$ and $\qdif<\qH$.  Any
  other equilibrium with (properly) mixed filter's strategy is
  utility-equivalent to $\blkprof$.
\end{proposition}




While Proposition \ref{prop:semi-regimes} and \ref{prop:semi-mixed}
establish that in a particular parameter regime, there are three
equilibrium profiles that earn payoffs higher than $V_i(\blkprof)$, if
one of these profiles Pareto dominated the others, then equilibrium
selection would be trivial.  Unfortunately, in general these profiles
are pairwise Pareto-incomparable.  Figure \ref{fig:incomp} illustrates
one such case, when $\Vfil(\fwdprof)>\Vfil(\mixprof)>\Vfil(\difprof)$
and $\Vcon(\fwdprof)<\Vcon(\mixprof)<\Vcon(\difprof)$.  Thus,
equilibrium selection cannot always be resolved via Pareto
optimality. We do not claim to resolve equilibrium selection but
instead,  characterize \vot \emph{at a particular profile} that in
some regimes is an equilibrium.

\begin{figure}
  \includegraphics[width=.8\textwidth]{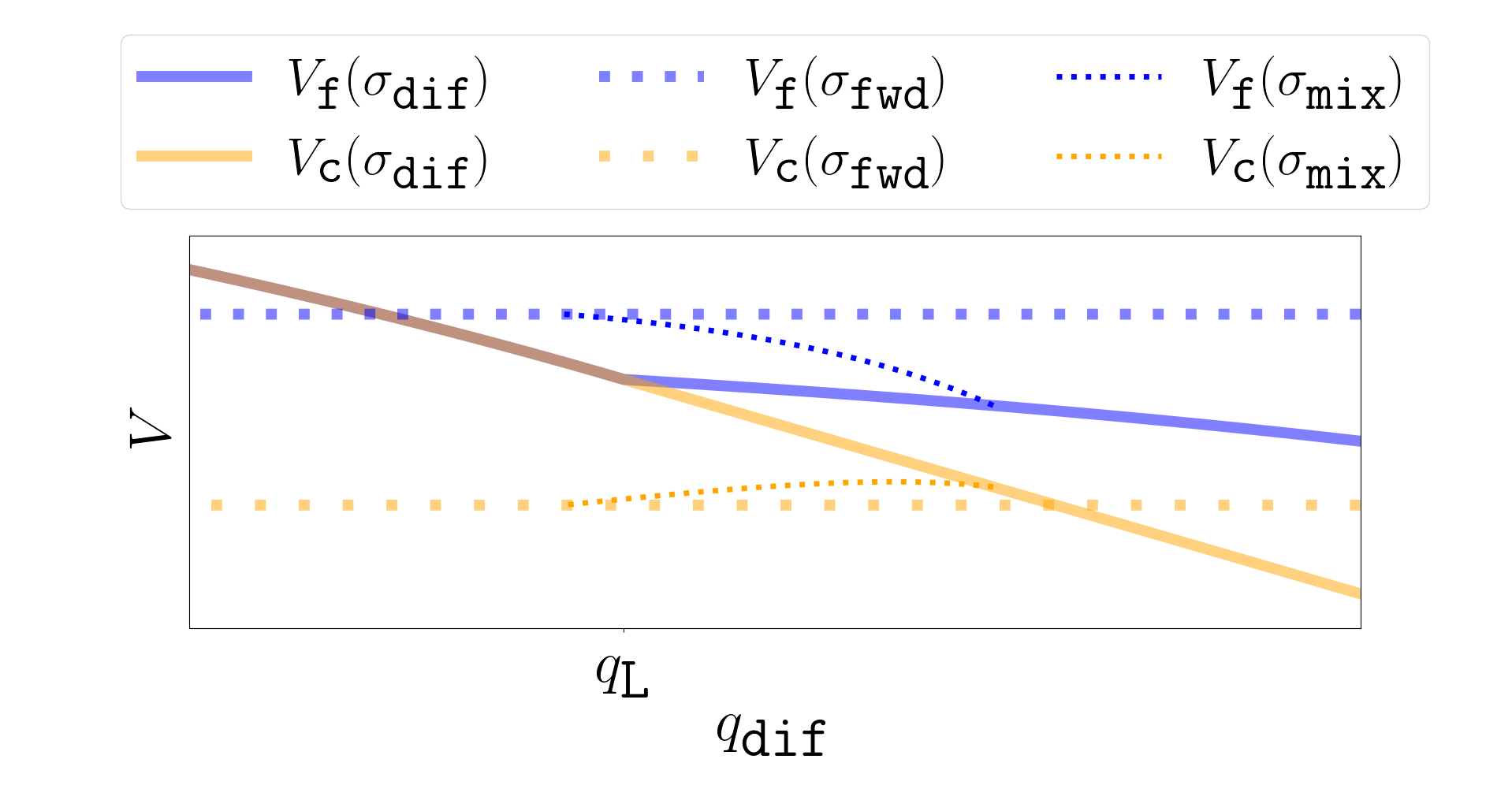}
  \caption{Payoffs under different profiles.  The payoffs under the
    mixed equilibria are only plotted in the region where such
    an equilibrium exists.  The parameters that generated the figure are
    $\lambda=2, b=1, c_1=1, c_2=2, q=.5, \pi_2=.2$ and $\pi_1$ ranges
    from $.2$ to $.8$.}
  \label{fig:incomp}
\end{figure}

\begin{theorem}\label{thm:semi-VoT}
~
\begin{itemize}

\item[(a)] \emthm{Zero \vot}
Under $\fwdprof$, we have
    $\partialpi{x}V_i(\fwdprof)=0$
for any $i\in\cbr{\fil,\con}$ and $x\in\bin$.

\item[(b)] \emthm{Filter Positive and Piecewise Constant \vot.}
Under $\difprof$, there is a constant \vot for the filter  when
$\qdif>\qL$ (resp., when $\qdif<\qL$).

\item[(c)] \emthm{Consumer Positive and Non-linear \vot.}
Under $\difprof$, there is a  positive and non-linear \vot for the consumer
     $\partialV{\Vcon}$
that satisfies \refeq{eq:thm:dvdp-nonlinear} with $\Vopt = \Vcon(\difprof)$.

\end{itemize}
\end{theorem}

Theorem \ref{thm:semi-VoT} says that improving the technology helps
both the consumer and the filter if the players follow $\difprof$ and
zero otherwise, just as in the case of aligned utilities.  However,
this does not give the complete picture of the \vot since any increase
in filter quality that allows the player to jump from an inefficient
to an efficient equilibrium, as shown in
Theorem \ref{thm:misaligned-escape}.   Put succinctly,
characterizing \vot without perfect alignment requires carefil
consideration of both the ``within profile'' \vot as provided in
Theorem \ref{thm:semi-VoT} and the ``between profile'' \vot established
in Corollary \ref{thm:misaligned-escape}.


For completeness, we also establish the \vot at the mixed equilibrium
profile.  As the following proposition shows, improving the technology
always helps the filter, but may hurt the consumer in the mixed equilibrium.

\begin{proposition}\label{prop:semi-mixprof}
Consider the parameter regime from Proposition~\ref{prop:semi-regimes}(c) and the equilibrium $\mixprof$ from Proposition~\ref{prop:semi-mixed}. The filter's \vot
     $\partialV{\Vfil}$
is positive. However, the consumer's \vot
     $\partialV{\Vcon}$
can be positive or negative, depending on the parameters.
\end{proposition}

The reason for the negative consumer  \vot  is that, as the
filter increases in quality, it must block more content to remain
indifferent when it receives a signal that content is malicious. While
this may help the consumer by blocking more malicious content, it also
blocks more legitimate content, and thus the  negative impact might outweigh
the positive effect.

Finally, while we do not claim to fully resolve equilibrium selection,
under the assumption that equilibrium selection is resolved by
choosing the equilibrium that is Pareto optimal among all equilibria
--- when it exists --- we can partially characterize the \vot for some
parameter regimes.

\begin{proposition}\label{prop:VoT-partial}
  Suppose $\qL<\qdif<q<\qH$ and equilibrium selection is resolved by
  selecting the equilibrium that Pareteo dominates any other
  \emph{equilibrium}.  Then if $\Lambda < \frac{\pi_0}{\pi_1}$,
  $VoT>0$ for both players.  If $\Lambda >
  \frac{\pi_0(1-\pi_1)}{\pi_1(1-\pi_0)}$, VoT is $0$ for both players.
\end{proposition}

Although far from a complete characterization, Proposition
\ref{prop:VoT-partial} pins down the \vot for a limited parameter
regime.

\end{prewebconf}

\section{Endogenous Attacker}
\label{sec:endog}
In this section, we extend our model to include the \emph{attacker}: a
third strategic player who is responsible for choosing the rate of
malicious content, $\rho_0$. We focus on \emph{aligned utilities}, to
better isolate the novelty brought by endogenizing the attacker. We
find two surprising consequences:  the consumer does not incur information costs in  equilibrium, and that improving the filter can make both the filter and consumer worse off.

\xhdr{Modeling choices and notation.}
We restrict the attacker to  pure strategies, \ie to choose its rate $\rho_0$ deterministically for the entire batch. One interpretation is that the attacker is not sophisticated enough to implement mixed strategies in this context.%
\footnote{A mixed strategy chooses $\rho_0$ at random once and keeps it fixed.} 



 As in the original model, all three players choose a strategy
  to use on the entire batch. 
  The attacker and filter move first and simultaneously. The consumer
  observes the attacker's choice of $\rho_0$ but not the filter's
  chosen strategy, and moves next. Therefore, for a fixed and known
  value of $\rho_0$, the game reduces to the \ourGame defined in
  Section \ref{sec:model}. Importantly, as per
  Remark~\ref{rem:model-order}, our results carry over to the variant
  where the consumer observes the strategies of both the attacker and
  the filter.  Furthermore, the results carry over to the case where the filter 
  also observes $\rho_0$.   




The attacker's expected utility, denoted $\Vatt$, is the expected number of malicious pieces of content that are accepted by
the consumer.%
\footnote{We do not impose \emph{production costs} on the attacker for generating malicious content. These costs are often small in practice:
\eg a generative AI model can produce many deep-fakes,
an inexpensive phish-kit can generate many fake emails \cite{kit}.
Our results generalize to allow for small but positive production costs.}
Fixing the strategies of all players and letting $Y$ be the number of malicious messages in a batch, we have
\begin{align}\label{eq:endog-Va}
\Vatt 
    = \E[Y]\;\Pr\sbr{\afil=\acon=1 \mid X=0},
\end{align}
where $\E[Y] = \rho_0 = q/(1-q)$.





Denote strategy profiles as
$(\rho_0,\mixsigma)$. We denote the players' utilities by $\Vatt$ and
$V = \Vfil = \Vcon$.   In general,
we expand any quantities that take as an input $\mixsigma$ to also
take as an input $\rho_0$.  For example, we write
$\Vatt = \Vatt(\rho_0,\mixsigma)$ and $V = V(\rho_0,\mixsigma)$.
Likewise, we write $q(\mixstrfil) = q(\rho_0,\mixstrfil)$ in
\refeq{eq:belief-defn}.

Note that the rate $\rho_0$ only enters the model through its impact on
     $q := \Pr[X=0] = \rho_0/(\rho_0+1)$,
which can take an arbitrary value in the interval $(0,1)$. Therefore, one could equivalently reparameterize the model so that the attacker sets $q\in (0,1)$ directly.


\newcommand{\stable}{Socially Optimal\xspace}
\newcommand{\optimal}{FC-optimal\xspace}


\xhdr{Equilibrium information costs.}
Our first result is that the consumer never incurs information costs in an equilibrium.

\begin{theorem}\label{thm:eattack-zero}
$\infoC(\rho_0^*,\mixsigma^*)=0$
for any equilibrium $(\rho_0^*,\mixsigma^*)$.
\end{theorem}

The key driver of
Theorem \ref{thm:eattack-zero} is that for a fixed filter's strategy,
the attacker's expected payoff under the consumer's best response is \emph{decreasing} in $\rho_0$ when
$q(\rho_0,\mixstrfil) \in (\qL,\qH)$.  Behaviorally, as the relative
proportion of malicious content rises, a combination of the consumer's
increased information costs and required certainty to accept content
\emph{reduces} the total amount of malicious content that is
ultimately accepted (of course, this comes at a higher cost due to
ignoring clean content).  On the other hand, for
$q(\rho_0,\mixstrfil)<\qL$, the attacker's payoff is increasing in $\rho_0$
since the consumer's best response is to accept all content.  As a
result, for a fixed filter strategy, the attacker's optimal strategy
is to set $\rho_0$ such that $q(\rho_0,\mixstrfil)=\qL$.  In this sense, the
consumer's attention serves as a deterrent to attack: the amount of
malicious content will not exceed the amount such that the consumer
incurs information costs in deciding whether content is legitimate.

\xhdr{Negative \vot.}
We find that improving the filter can \emph{reduce}
the equilibrium utility of the filter and the consumer.

As in Section~\ref{sec:aligned}, we focus on equilibria
$(\rho_0^*, \mixsigma^*)$ that maximizes the utility for the filter and the consumer,
\ie satisfy
\begin{align*}
    V(\rho_0^*, \mixsigma^*)\geq V(\rho_0, \mixsigma)
\text{ for any  equilibrium $(\rho_0,\mixsigma)$}.
\end{align*}
and label this equilibrium payoff $\Vopt$.
We are interested in \vot in terms of $\Vopt$.



Our negative \vot result can now be succinctly formulated using the
ratio $\frac{\pi_0}{\pi_1}$ and the threshold  $\Lambda$ from
\refeq{eq:lambda-defn}.

\begin{theorem}[Negative \vot] \label{thm:eattack}
Suppose $\pi_0/\pi_1<\Lambda$.  Then sufficiently improving both $\pi_0$ and $\pi_1$ strictly decreases the equilibrium utility $\Vopt$. More formally: there exist thresholds $\hat{\pi}_0\in(\pi_0,1)$ and $\hat{\pi}_1\in (0,\pi_1)$ such that
for any $\pi'_0\in(\hat{\pi}_0,1)$ and $\pi'_1\in(0,\hat{\pi_1})$ improving the filter quality to $(\pi'_0,\pi'_1)$ strictly decreases $\Vopt$.
\end{theorem}



What drives this result is that improvements in filter technology
can be completely crowded out by an increase in the attack propensity.
One key reason is that the socially optimal equilibrium switches from $\fwdprof$ to $\difprof$ as the filter technology improves.
Specifically, when the filter is poor quality, the socially optimal
equilibrium is $\fwdprof$.  Then, the attacker sets $\rho_0$
such that $q(\rho_0, \fwdprof)=\qL$ and
the consumer accepts all content.  However, for a high quality filter,
the socially optimal equilibrium is $\difprof$.  In that case, the
attacker sets $\rho_0$ such that $q(\rho_0, \difprof) = \qL$.  Consequently, the
expected fraction of malicious content that reaches the consumer is the
same in both equilibria and therefore, the equilibrium expected utility
for the filter and consumer \emph{conditional on content reaching the
  consumer} is the same. However, since
under $\difprof$ the filter blocks some clean content, the filter's and consumer's
expected utility under the $\difprof$ is strictly lower than the
expected utility under $\fwdprof$.  Although under $\difprof$ the filter blocks some malicious content, that benefit is not justified by the increase in attack intensity.

Another key feature driving this result is the
filter's inability to commit to $\strfwd$.  If the filter
were able to commit to $\strfwd$, then equilibrium expected utilities
would not depend on $\pi_0$ and $\pi_1$ and thus payoffs would not
change as the filter improved in quality.  However, because the filter
and attacker act simultaneously, once the filter is of sufficiently high quality, the
filter has an incentive to switch to $\strdif$.  However, under
$\strdif$, the attacker increases $\rho_0$, ultimately lowering
equilibrium expected payoffs for the filter and consumer.

\begin{prewebconf}
\begin{theorem}[Negative \vot] \label{thm:eattack}
For some parameter regime, a coordinate-wise increase in
$(\pi_0,-\pi_1)$ reduces $\Vopt$.
\end{theorem}


The intuition for this results is as follows.  In an \stable
equilibrium, the attacker will always attack right up to the point in
which any higher $q$ would induce the consumer to incur information
costs.  That is, for a fixed filter strategy, the optimal strategy for
the attacker (assuming the consumer observes its action and best
responds), is to set $q$ as high as possible such that the consumer's
optimal strategy is to not inspect the content and simply accept all
content.  Specifically, if the filter chooses $\strfwd$ the attacker's
optimal strategy is to set $q=\qL$.  If the filter chooses $\strdif$,
the attacker's optimal strategy is to choose $q$ such that
$\qdif=\qL$.  We denote this value of $q$ as $\qdif^{-1}$.  This is
not terribly surprising since consumer inspection reduces the
probability that malicious content is accepted.  Knowing this, the
only two possible \stable equilibria are if the filter either forwards
all content or plays $\strdif$. In either case the consumer is faced
with the same prior probability that content is malicious upon
receiving content.  However, it is imperative to not interpret this
result as saying information costs are irrelevant.  On the contrary,
the information costs and assoicated parameters \emph{determine} the
attacker's optimal strategy.

By our notion of stability, \emph{only} $\strfwd$ or $\strdif$ can be
\stable equilibrium strategies for the filter.  Additionally since the
attacker chooses the maximal $q$ such that the consumer is no longer
willing to incur information costs, any \stable equilibrium profile
has the consumer accepting all content.  That means the only
(attacker, filter) strategy pairs supported in equilibrium are
$(q_L, \strfwd)$, $(\qdif^{-1}, \strdif)$.  The consumer and filter
prefer $(q_L, \strfwd)$ over $(\qdif^{-1}, \strdif)$ since under
$\strfwd$, they do not incur costs for the filter blocking clean
content.  However, it might be the case that $\strfwd$ is not
equilibrium (stable or otherwise).  That is, for $q=\qL$, if the
filter has sufficiently high quality, $\strfwd$ it is not a best
response to $q=\qL$ since it could potentially block a high proportion
of malicious content.  Knowing this, the attacker would not choose to
set $q=\qL$ and thus the only \stable equilibrium is the one in which
the filter choose $\strdif$.

\bigskip

\jgcomment{Delete the rest of this section}
In the remainder of this section, we provide details to
Theorem~\ref{thm:eattack} and Theorem~\ref{thm:eattack-zero}. Along
the way, we partially characterize the equilibrium selection and the
value of technology.




\begin{proposition}\label{prop:attacker-best-eq}
  The only two possible \stable equilibria not equivalent to
  $ \blkprof$ are $(\qL,\fwdprof)$ or
  $(\qdif^{-1};\,\difprof)$.  Furthermore, if either of those are
  \stable equilibria the consumer's equilibrium
  strategy is to accept all content.
\end{proposition}


Proposition \ref{prop:attacker-best-eq} pinpoints the only two
possible \stable equilibrium profiles not equivalent to a trivial
regime in which the consumer or filter simply block all content.  We
an now characterize when each of the profiles are \stable equilibria.

 \begin{proposition}{\label{prop:endeq}}[Equilibrium Characterization]
       Assume $\qdif^{-1}>\qH$, then
       \begin{enumerate}
       \item $(\qdif^{-1},\,\difprof)$ is a \stable equilibrium.
       \item $(\qL,\fwdprof)$ is a stable equilibrium if and only if
       $\Lambda > \frac{\pi_0}{\pi_1}$.
     \end{enumerate}
 \end{proposition}

 Proposition \ref{prop:endeq} says that when $\qdif^{-1}>\qH$,
 $(\qdif^{-1},\,\difprof)$ is always a \stable (but not necessatily
 \optimal) equilibrium.  However, $(\qL,\fwdprof)$ is only a \stable
 regime in certain parameter regimes.  The following proposition
 establishes regions whereach of those are optimal equilibria.

 \begin{proposition}\label{prop:endog-regimes}
  Assume $\qdif^{-1}>\qH$.  Then
\begin{itemize}
\item[(a)] If $\Lambda<\frac{\pi_0}{\pi_1}$, then
  $(\qdif^{-1}, \difprof)$ is an  \optimal equilibrium.

\item[(b)] If
    $\Lambda > \frac{\pi_0(1-\pi_1)}{\pi_1(1-\pi_0)}$
    then $(\qL,\fwdprof)$ is an \optimal equilibrium.
    \end{itemize}
\end{proposition}

Proposition \ref{prop:endog-regimes} says there exists two
non-overlaping regimes in which each of the profiles are the strictly
optimal equilibrium.  The following result establishes the consumer's
and filter's preference over any possible equilibria

\begin{proposition}[\label{prop:fwdbetter}]
  $\Vnorm(\qL,\fwdprof)$, which does not depend on $\pi_0$ and $\pi_1$
  achieves higher payoffs  then $\Vnorm(\qdif = \qL;\,\difprof)$ for any
  values of $\pi_0, \pi_1 \in(0,1)$.
\end{proposition}

Proposition \ref{prop:fwdbetter} says that the filter and consumer
always prefer the equilibrium where the attacker plays $\qL$ than any
other equilibrium for any values of the parameters.  Taken together
with proposition \ref{prop:endog-regimes}, anytime the filter's
quality changes such that $(\qL,\fwdprof)$ goes from being a \stable
equilibrium to not existing as an equilibrium, there will be a
negative \vot.

We can finally show that this happens though increasing filter
quality.

\begin{proposition}\label{prop:eattack-details}
Fix all model parameters other than $\pi_0,\pi_1$.
Suppose $\pi_0,\pi_1$ change so that we transition from regime (b) to
regime (a) in Proposition~\ref{prop:endog-regimes}. Then the \optimal
utility $\Vopt$ strictly decreases. However, such transition can
happen by increasing $(\pi_0,-\pi_1)$ coordinate-wise, corresponding
to an unambiguous improvement in filter quality.
\end{proposition}


For completeness, Since there are regimes in which both
$(\qdif^{-1}, \difprof)$ and $(\qL,\fwdprof)$ may be optimal
equilibria, it would be insightful to characterize the \vot in those
regimes.  We do so in the following propostion:

\begin{proposition}\label{prop:endog-VoT}~
\begin{itemize}
\item[(a)] For strategy profile $(\qL,\fwdprof)$, \vot is
$\partialpi{0}V = \partialpi{1}V = 0 $.

\item[(b)] For strategy profile $(\qdif^{-1};\,\difprof)$, the \vot is
    $\partialpi{0}V = 0$,
whereas
    $\partialpi{1}V<0$ is constant in $(\pi_0,\pi_1)$.
\end{itemize}
\end{proposition}
\end{prewebconf}





\section{Conclusions and Open Questions}
\label{sec:future}

We develop a model of strategic interactions between a content filter
and inattentive content consumers; such interactions are a common
feature in many applications. 
Our equilibrium analysis undermines the common notions that improving
filter quality is unambiguously beneficial and that the improvements are
necessarily linear in the natural parameters (such as the true/false
positive rates). We conclude that consumers' strategic
inattention is essential for the analysis of content filtering.




The main policy implication is that content filtering does not reduce
to a classification problem in machine learning. In addition to rote
improvements in filter quality, one should consider interventions to
reduce consumers' information costs and increase vigilance.%
\footnote{Such interventions are not uncommon in practice. Mandatory
  corporate trainings are now wide-spread. Some IT departments even
  implement ``secret exercises", \eg send out phishing emails to all
  employees and reprimand those who fall for these emails.}  Our
analysis illuminates non-obvious positive consequences of these
interventions that arise due to strategic interactions: \eg increasing
the marginal benefits of improvements in filter quality, or
disincentivizing the attacker from inserting more malicious
content. Detailing whether and which interventions are desirable remains an intriguing open
question.




We focus on a homogeneous and stationary world in which the homogeneous players' strategies are
non-adaptive and fixed throughout. Effectively, we consider a
  ``single-round'' game that concerns a single piece of
   content. 
  This stationary world is, of course, an idealization of
a dynamic world in which the players continuously adapt to one
another. Such dynamic worlds are notoriously difficult to analyze, and
are not well-understood even in simple scenarios.%
\footnote{They are studied in (decentralized) multi-agent learning, \eg Ch. 9.5 in
  \cite{slivkins-MABbook} for  introductory background.}  Focusing
on equilibria of a ``single-round'' game is a common route towards
tractability.  Nevertheless, adding dynamics with heterogeneous consumers is a viable extension.  


A key simplification in our model is that all legitimacy-related
quantities are \emph{binary}: the legitimacy itself, the filter's signal and action and the consumer's signal and action.
Indeed, the
filter's and the consumer's signal could be fractional, reflecting the
likelihood of the content piece being malicious. Filter's actions
could also
include, \eg putting the content piece into a spam folder
or attaching a warning.
Furthermore, the content piece itself may sometimes be a mix
of genuine and malicious, \eg a genuine social media post may be
contaminated by propaganda. Accordingly, a consumer might choose an
`intermediate' action, \eg accept the content piece with some
reservations.
Relaxing these binary choices could
potentially lead to more refined conclusions,
but might also lose the appealing simplicity and tractability of
the ``binary'' model.



Our model of information costs, while suitable (and standard) for
idealized models, could potentially be refined to reflect more
realistic scenarios of information discovery. First, the
\emph{process} of information discovery could be modeled more
explicitly, perhaps via an analogy to machine learning algorithms for
similar problems. Second, the \emph{information sources} available to
a human user may differ from the one readily available to the
filter. For example, a human receiving an email might intuitively pick
up on a suspicious tone or an unusual visual layout, whereas a spam
filter would be restricted to specific pre-trained characteristics of
the email. Moreover, a human user might do a quick web search to
resolve a suspicion (\eg of spam, phishing, or misinformation), or
even ask a friend, whereas a spam/content filter might consult its
internal database. On the other hand, such refinements might be
application-specific and/or involve some unobvious modeling choices.

Another approach towards modeling information costs is to handle a
large, abstract \emph{class} thereof, without attempting to
micro-found any particular function shape in this class.
In \GenCostsLocation, we obtain an initial result in this direction, generalizing the conclusions in Section~\ref{sec:aligned} to arbitrary information costs under
some generic conditions.

\newpage

\bibliography{refs,bib-abbrv,bib-slivkins}


\appendix

\section{Generalized information costs}
\label{sec:gen}

This appendix begins to generalize our notion of information costs.
Specifically, we consider convex/concave information costs (defined below). We focus on aligned utilities, and we restrict the filter to only use pure action strategies. The ``interesting" parameter regime here is when the differentiating profile $\difprof$ is socially optimal and the consumer does not incur information costs.
\footnote{Indeed, the alternatives are essentially trivial: either the filter does not choose the differentiating strategy $\strdif$ (and players' payoffs do not depend on the filter quality), or the consumer does not incur information costs and equilibrium payoffs are linear in filter quality.}
When and if this parameter regime occurs, we show that the players' utility is strictly increasing in the filter quality.



Let us formulate our cost model. Let $C(\mu; q(\difprof))$ be the consumer's cost for choosing information strategy
    $\mu=\rbr{\tilde{\pi}_0, \tilde{\pi}_1}$
  when their prior belief that content is
  malicious is $q(\mixstrfil)$. We assume the following:
\begin{enumerate}
\item  $C$ is strictly convex in $\mu$;
\item $C$ is strictly concave in $q(\mixstrfil)$;
\item $ C$ is differentiable in $\mu$ and $q(\mixstrfil)$;
\item $C(\mu;q(\mixstrfil))=0$ if and only if
      $P(X|\Psi_{\texttt{c}})=P(X)$.
\end{enumerate}

\noindent Assumptions 1-3 are standard.  Assumption $4$ says that the
consumer can costlessly choose to gather no information and any
other information that shifts the distribution away from the prior must be costly.

The main result of this section is stated as follows.

\begin{prop}\label{prop:gen-costs}
Consider information costs that satisfy assumptions 1-4 above. Posit aligned utilities. Restrict the filter to only use pure action strategies.

Suppose the parameters are such that in an open neighborhood around $(\pi_0,\pi_1)$ (and fixing the other parameters) $\difprof$ is socially optimal and the consumer incurs positive information costs under $\difprof$. Then
\begin{align}\label{eq:prop:gen-costs}
 \frac{\partial V(\difprof)}{\partial \pi_1}>0 \text{    and    } \frac{\partial V(\difprof)}{\partial \pi_2}<0.
\end{align}
\end{prop}

We prove this propositions in what follows. We make the assumptions therein without further notice.

Let us adopt the following notation:
\begin{align}
\vec{q}(\mixstrfil)
    &= [q(\mixstrfil), (1-q(\mixstrfil))]^{\prime}\nonumber \\
\tilde{\Pi}(\mu)
    &= \begin{bmatrix}
        \tilde{\pi}_1 & 1-\tilde{\pi}_1 \\
        \tilde{\pi}_2 & 1-\tilde{\pi}_2
    \end{bmatrix} \nonumber \\
\vec{R}
    &= \begin{bmatrix}
            0 & -c_1 \\
            c_2 & b
    \end{bmatrix}
    \label{eq:GenCosts-notation}
\end{align}

Under strategy profile $\difprof$, the consumer chooses information strategy
  $\mu=\rbr{\tilde{\pi}_0, \tilde{\pi}_1}$
and pure action strategy $s_{\texttt{c}}$ in order to optimize
\begin{equation*}
  V_g(\difprof, s_{\texttt{c}}, \mu)
  =  \mathrm{tr}\left(\tilde{\Pi}(\mu)\vec{R}\right) \cdot \vec{q}(\difprof) - C(\mu; q(\mixstrfil)).
\end{equation*}

The consumer has a unique optimal choice of $\mu$. This is because the consumer's strategy space is compact and $V_g$ is a concave function minus a convex function.




We prove the first statement in \eqref{eq:prop:gen-costs}, the second case follows similarly
  \begin{align*}
    & \frac{\partial V(\sigma^*)}{\partial \pi_1} =  \nonumber \\
    & A\bigg(\left(\frac{\partial Z}{\partial \tilde{\pi}_1} +
    \frac{\partial Z}{\partial \tilde{\pi}_2}\right) \times
  \left(\frac{\partial \tilde{\pi}_1}{\partial q(\difprof)} \frac{\partial q(\difprof)}{\partial \pi_1}\right) + \nonumber \\
  & \frac{\partial Z}{\partial q(\difprof)}\frac{\partial q(\difprof)}{\partial\pi_1}\bigg) + Zq
  \end{align*}
  where
  \begin{align*}
  A &= \left((1-\pi_1)q + (1-\pi_2)(1-q)\right) \\
   Z &=(
    (1-\tilde{\pi}_1)q(\difprof)(-c_2) + \\
    &(1-q(\difprof))(\tilde{\pi}_2(-c_1) +
    (1-q(\difprof))(1-\tilde{\pi}_2)b + \\
    & C(\mu; q(\difprof)))
    \end{align*}
  However, by enforcing optimality for the consumer $\frac{\partial Z}{\partial \tilde{\pi}_1}=
  \frac{\partial Z}{\partial \tilde{\pi}_2}=0$.  Therefore, it is now sufficient to show that
  \begin{align}
    (1-\tilde{\pi}_1)c_2 + (1-q(\difprof))\frac{\partial C}{\partial q(\difprof)} + C(\mu;q(\difprof))>0
    \label{eq:fin}
  \end{align}
  By the definition of concavity (and notationally dropping the
  dependence on $\mu$), it must be that
  \begin{align*}
  C(q(\difprof)') - C(q(\difprof)) < \frac{\partial{C}}{\partial q(\difprof)}(q(\difprof)' - q(\difprof))
\end{align*}
for any $q(\difprof)$ and $q(\difprof)^{\prime}$. Plugging in $1$ for
$q(\difprof)^{\prime}$ and rearranging yields
\begin{align*}
  \frac{\partial{C}}{\partial q(\difprof)}(1 - q(\difprof)) +C(q(\difprof))>C(1)>0
  \end{align*}
which then implies that inequality \eqref{eq:fin} is satisfied since the
first term is positive and the sum of the second and third term are
positive and thus completes the proof.


\section{Proofs\vspace{2mm}}
\label{sec:proofs}


\begin{prop}\label{prop:unreason}
    The ``unreasonable''' strategy cannot be an equilibrium in which the consumer doesn't block all content.  
\end{prop}
\begin{proof}
Suppose the unreasonable strategy were played in an equilibrium in which the consumer accepted content with a positive probability.  . Then the filter would strictly profit by replacing its strategy with the mixed strategy $\pi_1  \strfwd + (1 - \pi_1) \strblk$. This yields strictly more utility because  under the unreasonable strategy, malicious content is forwarded at rate $\pi_0$ and genuine content is forwarded at rate $\pi_1$. Under the deviation, all content is forwarded at rate $\pi_1$. This deviation strictly increases content utility  and weakly reduces information costs.
\end{proof}
 
\subsection*{From Section~\ref{sec:trivial}: Consumer Beliefs\vspace{2mm}}

\begin{lemma}
\label{lem:cbr}
Fix a particular $\mixstrfil$, mixed strategy of the filter. The consumer's unique best response to it is
\begin{align}
  \tilde{\pi}_0^*  &= \frac{(1-\mathcal{P}_c)}{(1-\mathcal{P}_c) +
                    \mathcal{P}_{c}e^{-c_2 / \lambda}}\nonumber\\
 \tilde{\pi}_1^* & = \frac{(1-\mathcal{P}_{c})e^{-c_1/\lambda}}{(1-\mathcal{P}_c) e^{-c_1/\lambda} + \mathcal{P}_ce^{b/\lambda}},                                \label{eqn:cbr}
\end{align}
where
\begin{align*}
\mathcal{P}_c &= \min(1,\, \max(0, \widetilde{\mathcal{P}}_c)) \\
\widetilde{\mathcal{P}}_c
    &=
  \frac{e^{b/\lambda}(\qfunc(\mixstrfil)-1) +
  e^{-c_1/\lambda}(1-\qfunc(\mixstrfil)e^{-c_2/\lambda})}{(e^{-c_2/\lambda}-1)(e^{b/\lambda}
  - e^{-c_1/\lambda})}.
\end{align*}


Furthermore, if $\mathcal{P}_c\in(0,1)$, the unique consumer optimal
actions is $s_{\texttt{c}}^*(0) = 0$ and $s_{\texttt{c}}^*(1) = 1$. If
$\mathcal{P}_c\notin(0,1)$, the consumer's optimal action is to ignore
all content or accept all content and choose a non-informative but
costless information strategy (i.e.
$\tilde{\pi}_1^* = \tilde{\pi}_2^*$).
\end{lemma}

\begin{proof}
  The consumer's strategy only impacts
  payoffs if $a_{\texttt{f}}=1$ so it sufficies to examine the
  consumer's best response conditional on the filter forwarding
  content.  Conditional on $a_{\texttt{f}}=1$, the consumer's decision
  problem is equivalent to a discrete choice problem under rational
  inattention as in \cite{mm} with prior probability
  $\qfunc(\mixstrfil)$.  Applying equation their equation 13 as well as the
  guarantees of uniqueness provided in \cite{mm} gives the claimed
  results.
\end{proof}

\begin{proof}[Proof of Lemma \ref{lem:belief-BR}]
$\mixstrfil$ only enters  \refeq{eqn:cbr} via $\qfunc(\mixstrfil)$.
\end{proof}

\begin{proof}[Proof of Proposition \ref{prop:goalposts}]
The proof immediately  follows from Lemma~\ref{lem:cbr}.
\end{proof}




\begin{proof}[Proof of Proposition \ref{cor:scaff-VoT}]
  To prove (a), under $\blkprof$ payoffs are $-c_1$ and do not
  depend on $\pi_0$ or $\pi_1$.  Under $\fwdprof$, payoffs are given
  by
  $\frac{1}{1-q}\E_{X,\scon}\sbr{ u(\acon,1,X) - \infoC\sbr{\scon;X \mid \afil=1}}$
  which does not depend on $\pi_0$ or $\pi_1$ directly nor through the
  consumer's optimal strategy since $\qfunc(\fwdprof)=q$ for any
  values of $\pi_0$ and $\pi_1$.  Finally, if $\qfunc(\difprof)>q_H$,
  the consumer's optimal strategy is to block all content and payoffs
  are again $-c_1$ and do not depend on $\pi_0$ and $\pi_1$.

  To prove (b), the consumer's strategy under $\difprof$ is to accept all
  content and thus \emph{per content} payoffs are given by
  $(1-q)(\pi_1(-c_1)+(1-\pi_1)b) + q(1-\pi_0)(-c_2)$ of which taking
  the derivatives and multiplying by $\frac{1}{1-q}$ are straightforward.
  \end{proof}

\subsection*{From Section~\ref{sec:aligned}: Aligned Utilities\vspace{2mm}}

\begin{lemma}\label{lem:blkbad}
 $V_i(\blkprof)\leq \min\rbr{V_i(\difprof), V_i(\fwdprof)}$,  where $i\in\cbr{\fil,\con}$.
\end{lemma}
\begin{proof}
  $\Vfil(\mixsigma)\geq \Vcon(\mixsigma)$ and since the consumer can
  always block all content for any $\mixstrfil$,
  $\Vcon(\blkprof)<\Vcon(\mixsigma)$ for any consumer-optimal
  $\mixsigma$.
  \end{proof}

\begin{lemma}
    \label{lem:terms}
For a single piece of content, let
    $ \widehat{u}_a(\qhat) = \E\sbr{u(a,\hat{X})}$
be the expected action payoff for aggregate action $a\in\bin$ if the content type $\hat{X}\in\bin$ is a random variable with $\qhat = \Pr[\hat{X}=0]$.
If $\qhat := \qfunc(\mixstrfil) \in (\qL,\qH)$ then
\begin{align*}
\Vind(\qhat)
    &=  \widehat{u}_1(\qhat) +\infoL\cdot\KL{\qhat}{\qL},\\
        &\qquad\text{where } \widehat{u}_1(\qhat)=(1-\qhat)b -\qhat c_2,  \nonumber \\
    &=  \widehat{u}_0(\qhat) +\infoL\cdot \KL{\qhat}{\qH},\\
        &\qquad\text{where } \widehat{u}_0(\qhat)= -(1-\qhat)c_1.
        \nonumber
\end{align*}
\end{lemma}

\begin{proof}[Proof of Lemma \ref{lem:terms}]
  Let $\tilde{\pi}_0^*$ and $\tilde{\pi}_1^*$ represent the consumer's
  optimal attention strategy as given in lemma \ref{lem:cbr}.
  Then
\begin{align*}
    &\hat{V}(\qfunc(\mixstrfil))= \\
    &\qquad\qfunc(\mixstrfil) \left( 1-\pot \right)  \left( -c_2 \right) \\
    &\qquad+ \left( 1-\qfunc(\mixstrfil) \right) \ptt   \left( -c_1 \right) \\
    &\qquad+ \left( 1-\qfunc(\mixstrfil) \right)  \left( 1-\ptt \right) b \\
    &\qquad- \lambda [H(\mathcal{P}_c) -q(\mixstrfil)H(\pot) - (1-q(\mixstrfil))H(\ptt)].
  \end{align*}

  Plugging in optimal values of $\pot$ and $\ptt$ from Lemma
  \ref{lem:cbr}, separating the logs and recognizing that
  $1-\mathcal{P}_c = \qfunc(\mixstrfil)\pot +
  (1-\qfunc(\mixstrfil))\ptt$ and plugging in for $\mathcal{P}_c$ from
  Lemma \ref{lem:cbr} yields
  \begin{align*}
  \hat{V}(\qfunc(\mixstrfil))
  &=\lambda\big[
        \qfunc(\mixstrfil)\log\left(\frac{\qfunc(\mixstrfil)}{1-\qfunc(\mixstrfil)}\frac{1-e^{-c_2/\lambda}}{e^{b/\lambda} - e^{-c_1/\lambda}}\right)\\
   &\quad     + \log\left(\frac{(1-\qfunc(\mixstrfil))(e^{b/\lambda}
        - e^{-(c_1+c_2)/\lambda})}{1-e^{-c_2/\lambda}}\right)\big].
  \end{align*}

Separting out the $\qfunc(\mixstrfil)$ terms from within the logs and substituting
$\frac{e^{-b/\lambda}(1-q_L)}{e^{c_2/\lambda}q_L}$ for the first term
and $\frac{e^b}{1-q_L}$ for the second term gives the claimed result
in terms of $q_L$.  Making similar substitutions with $q_H$ gives the
second expression in the claimed result.
\end{proof}

\begin{lemma}[Profile Payoffs]
  \label{lem:payterm}
  If $q, \qfunc(\difprof) \in (\qL, \qH)$ then
  \begin{align}
    &V(\fwdprof) = \frac{1}{1-q}\hat{V}(q) \nonumber \\
    &V(\difprof) = -c_1 + \frac{1}{1-q}\beta\lambda\KL{q(\strdif)}{q_H} \label{eq:to_dif}
  \end{align}
\end{lemma}

\begin{proof}
This is a straightforward application of Lemma \ref{lem:terms}
\end{proof}




\begin{proof}[Proof of Proposition \ref{prop:diff}]
When $q(\mixstrfil)>\qH$ the consumer ignores all
content. When $q(\mixstrfil)<\qL$ the consumer accepts all content.
Otherwise, payoffs are given in Lemma \ref{lem:payterm}, the
proof follows by subtracting $V(\fwdprof)$ from
$V(\difprof)$ for each of the regimes.
\end{proof}

\begin{prewebconf}

  \begin{lemma}\label{lem:qhv}
  Let $\mixsigma^1$ and $\mixsigma^2$ be  consumer-optimal profiles where
  the filter plays a (possibly degenerate) mixed strategy
  $\mixstrfil^1$ and $\mixstrfil^2$, respectively.  If
  $\qL<q(\mixstrfil^1)<q(\mixstrfil^2)<\qH$ then $V(\mixsigma^1)>V(\mixsigma^2)$
  if and only if
  $\beta_1 \KL{q(\mixstrfil^1)}{\qH}-\beta_2\KL{q(\mixstrfil^2)}{\qH}>0$
  where the $\beta_i$  is the unconditional probability of the
  filter forwarding under $\mixstrfil^i$.
\end{lemma}
\begin{proof}
  This is again a straightforward application of lemma \ref{lem:terms}
  where the filter plays a differentiating strategy with signals
  governed by
  $\pi_0(\gamma^i)=\pi_0(\gamma^i_0-\gamma^i_1)+\gamma^i_1$ and
  $\pi_1(\gamma^i)=\pi_1(\gamma^i_0-\gamma^i_1)+\gamma^i_1$ and
  $\gamma_i$ are the mixing probabilities.
\end{proof}

    \begin{lemma}\label{lem:qhv2}

      Let $\mixsigma$ be a consumer-optimal profile where the filter
      plays a (possibly degenerate) mixed strategy.  If
      $q(\strdif)<\qL\leq q(\mixstrfil)<\qH$ then
      $V(\difprof)>V(\mixsigma)$ if and only if
      $\Delta U^{\prime}>\beta_{\gamma}\KL{q(\mixsigma)}{\qH}$ and
      $\Delta U^{\prime}=(1-q)(1-\pi_1)(b+c_1) -q (1-\pi_0)c_2$
      where $\beta_{\gamma}$ is the unconditional probability of the
      filter forwarding under $\mixstrfil$.

\end{lemma}
\begin{proof}
  Using the $\qH$ expression for $\hat{u}$ yields
       \begin{align*}
         V(\difprof) - V(\mixsigma) > 0\implies \nonumber \\
         (1-q)(\pi_1(-c_1)+(1-\pi_1)b) -qc_2(1-\pi_0) >-(1-q)\pi_1^{\gamma}c_1 + \beta_\gamma (-(1-q(\mixstrfil))c_1 + \lambda \KL{q}{q_H})
       \end{align*}
    where the $\pi_1^{\gamma}$ is as in lemma \ref{lem:qhv}.
    Rearranging the terms gives the claimed result.
    \end{proof}

\begin{proof}[Proof of Proposition \ref{prop:aligned-best-eq}].
  (This is a non-trivial proof because the filter's mixing
  probabilities enter non-linearly into payoffs since it impacts the
  information costs through the consumer's prior.  Therefore, standard
  arguments regarding mixed strategies  --- even when incentives are
  completely aligned ---  do not apply.)

  First note that $V(\difprof)$ and $V(\fwdprof)$ are lower bounded by
  $V(\blkprof)$ since the consumer can always reject all content.  It
  suffices to show that no mixed strategy profile does better than the
  maximum of $V(\difprof)$ and $V(\fwdprof)$.  Specifically, since
  this is a game of common interest, if no consumer-optimal mixed
  strategy does better than the greater of $V(\difprof)$ and
  $V(\fwdprof)$, than the profile yielding the higher payoff is the
  socially optimal equilibrium.  Lemma \ref{lem:cbr} already
  established that the consumer's best response is unique so any mixed
  strategy equilibrium is one in which only the filter mixes.

  Let $\mixstrfil$ be the filter's mixed strategy and assume the
  consumer best responds. Then the
  game is equivalent to a new game where the
  filter plays a differentiating strategy but its signal distribution
  is given by $\pi_0(\gamma_0-\gamma_1)+\gamma_1$ and
  $\pi_1(\gamma_0-\gamma_1)+\gamma_1$ where $\gamma_0,\gamma_1$ are
  the filter's random probabilities of blocking at each information
  set.  If $\gamma_0<\gamma_1$ then the filer is more likely to
  forward malicious content than clean content, which is a a weakly
  dominated strategy, so we can focus on the case where
  $\gamma_0>\gamma_1$.  This also implies that
  $q(\strdif)<q(\mixstrfil)<q$.

  The case where $q(\mixstrfil)>\qH$ is trivial since the optimal
  payoff under that regime is $V(\blkprof)$, which is weakly dominated
  by any other pure strategy profile.  Furthermore when
  $q(\mixstrfil)\leq\qL$, the payoffs are linear in $\gamma_0$ and
  $\gamma_1$ and thus the filter would be better on the edge, which
  corresponds either to $\fwdprof$ or $\difprof$ or setting
  $q(\mixstrfil)=\qL$. which we will show is also not optimal.


  Fix the filter's mixed strategy at $\mixstrfil$ such that
  $\qL\leq q(\mixstrfil)<\qH$.  Applying lemma \ref{lem:terms} means
  that the player's optimal payoff when the filter plays mixed
  strategy $\mixstrfil$ is given by
  $\beta_{\gamma}\lambda \KL{q(\mixstrfil)}{\qH}$ where
  $\beta_{\gamma}$ is the unconditional probability the filter
  forwards content.  Let
  $t=\frac{q-q(\mixstrfil)}{q-q(\strdif)}=1-\frac{\gamma_0}{\beta_{\gamma}}$
  such that $tq(\strdif) + (1-t)q = q(\mixstrfil)$.  Then by convexity
  of KL divergence
    \begin{align*}
      \KL{q(\mixstrfil)}{\qH}<t \KL{q(\strdif)}{\qH} + (1-t)\KL{q}{\qH} \nonumber \\
      \rightarrow \gamma_0\big(\KL{q(\strdif)}{\qH} - \KL{q}{\qH}\big) < \beta_{\gamma}\KL{q(\strdif)}{\qH} - \beta_{\gamma}\KL{q(\mixstrfil}{\qH} \nonumber
    \end{align*}
    writing the first $\beta_{\gamma}$ term as $\beta + \beta_{\gamma} - \beta$ and
    rearranging yields
    \begin{align*}
      (\gamma_0+\beta - \beta_{\gamma})\KL{q(\strdif)}{\qH} - \gamma_0 \KL{q}{\qH}< \beta \KL{q(\strdif)}{\qH} - \beta_{\gamma}\KL{q(\mixstrfil)}{\qH} \nonumber
    \end{align*}
    and recognizing that $\gamma_0+\beta -
    \beta_{\gamma}=(1+\gamma_0-\gamma_1)\beta$ the inequality can be
    written as
\begin{align*}
  (1+\gamma_0-\gamma_1)\beta \KL{q(\strdif)}{\qH} - \gamma_0 \KL{q}{\qH}< \beta \KL{q(\strdif)}{\qH} - \beta_{\gamma}\KL{q(\mixstrfil)}{\qH} 
  \rightarrow \gamma_0\big(\beta \KL{q(\strdif)}{\qH} -  \KL{q}{\qH}\big)< \beta \KL{q(\strdif)}{\qH} - \beta_{\gamma}\KL{q(\mixstrfil)}{\qH} 
\end{align*}
where the second line follows because $1-\gamma_1>0$.

\textbf{Case 1:} Suppose $\qL\leq q(\strdif)<q(\mixstrfil)<\qH$: First, suppose
$V(\difprof)>V(\fwdprof)$, then by lemma \ref{lem:qhv} it suffices to
show that
$\beta \KL{q(\strdif)}{\qH}-\beta_{\gamma} \KL{q(\mixstrfil)}{\qH}>0$.
By lemma \ref{lem:qhv} the left hand side of equation \ref{eq:nomix1}
is positive so the right hand side must also be positive thus proving
the sufficient condition.  Now, suppose $V(\difprof)<V(\fwdprof)$,
then again by lemma \ref{lem:qhv} it suffices to show that
$\KL{q}{\qH}-\beta_{\gamma} \KL{q(\mixstrfil)}{\qH}>0$.  This can be
accomplished by adding $\KL{q}{\qH}$ to each side and rearranging
\ref{eq:nomix1} to read
\begin{align*}
  (\gamma_0-1)(\beta \KL{q(\strdif)}{\qH} -\KL{q}{\qH})< \KL{q}{\qH} - \beta_{\gamma}\KL{q(\mixstrfil)}{\qH}
\end{align*}
and since by assumption $V(\difprof)<V(\fwdprof)$ both terms on the
left hand side are negative so the right hand side must be positive,
thus completing the proof.

\textbf{Case 2:} Suppose $q(\strdif)<\qL\leq q(\mixstrfil)$.  First
suppose $V(\difprof)>V(\fwdprof)$.  By lemma \ref{lem:qhv2} we must
show that $\Delta U^{\prime}>\beta_{\gamma}\KL{q(\mixstrfil)}{\qH}$.
By applying lemma \ref{lem:qhv2}, $\Delta U^{\prime}>\KL{q}{\qH}$ and
thus we can plug $\Delta U^{\prime}$ into \ref{eq:nomix00} and
rearrange to obtain.
\begin{align*}
  (\gamma_0 - \gamma_1)\beta \KL{q(\strdif)}{\qH} + \beta_{\gamma}\KL{q(\mixstrfil)}{\qH}<\gamma_0 \Delta U^{\prime} 
\end{align*}
Since by assumption $\gamma_0>\gamma_1$, and $\gamma_0<1$, expression
\ref{eq:nomix2} implies $\beta_{\gamma}\KL{q(\mixstrfil)}{\qH}< \Delta
U^{\prime}$, completing the proof.

Finally, suppose $V(\difprof)<V(\fwdprof)$. Then it we must show that
show that $\KL{q}{\qH}>\beta_{\gamma}\KL{q(\mixstrfil)}{\qH}$ which we
can do by rearranging equation \ref{eq:nomix00} to be
\begin{align*}
  (\gamma_0-\gamma_1)\beta \KL{q(\strdif)}{\qH} <  \gamma_0 \KL{q}{\qH} - \beta_{\gamma}\KL{q(\mixstrfil)}{\qH} 
\end{align*}
and since by assumption $(1<\gamma_0<\gamma_1)$, equation
\ref{eq:nomixlast} implies
$\KL{q}{\qH}>\beta_{\gamma}\KL{q(\mixstrfil)}{\qH}$.  This completes
the proof for all cases.
  \end{proof}

\end{prewebconf}

\begin{proof}[Proof of Theorem \ref{thm:dvdp}]
  Parts (a,b) 
  follow directly from Proposition~\ref{cor:scaff-VoT}.  
   For part (c), take the derivative of equation
  \refeq{eq:to_dif}.
\end{proof}

\begin{proof}[Proof of Proposition \ref{prop:aligned-best-eq}]
To show that no mixed strategy can yield higher utilities than
$\max(V(\difprof), V(\fwdprof)$, 
  proceed by contradiction.  Suppose there was a socially
  optimal profile in which the filter blocked with probability
  $\gamma_0$ when $\signal_{\fil}=0$ and blocks with probability
  $\gamma_1$ when $\signal_{\fil}=1$.  This is equivalent to a game in
  which the filter plays a differentiating strategy with the filter's
  signal distribution given by
  $\pi_i^{\prime} = \pi_i(\gamma_0 - \gamma_1) + \gamma_1$.  Denote
  that profile $\difprof^{\prime}$.  Suppose
  $q(\difprof^{\prime})\in(\qL, \qH)$ (the case where
  $q(\difprof^{\prime})<\qH$ is trivial and the case where
  $q(\difprof^{\prime})<\qL$ follows by analogy). Taking the total
  derivative of $\pi_i^{\prime}$ and setting them equal to $0$ says
  that $\pi_0$ is constant in a neighborhood around
  $\gamma_0, \gamma_1$ if
  $\frac{d \gamma_0}{d \gamma_1} = \frac{(1-\pi_0)}{\pi_0}$.  However,
  at that rate of marginal substitution, it must be that
  $\frac{d \gamma_0}{d \gamma_1} < \frac{(1-\pi_1)}{\pi_1}$ which by
  total differentiation implies $\pi_1$ is decreasing.  Since for a
  differentiating profile, expected utility is decreasing in $\pi_1$,
  the above implies that the filter can change $\gamma_0$ and
  $\gamma_1$ so that $\pi_0$ does not change, $\pi_1$ decreases, thus
  total expected utility increases.  This contradicts the initial
  assumption that $\difprof^{\prime}$ was socially optimal.
 \end{proof}



 \begin{proof}[Proof of Corollary \ref{prop:aligned-overcome}]
   As already established, $V(\difprof)$ is weakly increasing in
   $\pi_0, -\pi_1$ and $\qdif$ is decreasing in those probabilities
   while $V(\fwdprof)$ is constant.  Therefore, as $\pi_0\rightarrow
   1$ and $\pi_1 \rightarrow 0$, condition $d$ in proposition
   \ref{prop:diff} is guaranteed to be satisfied.
\end{proof}

\subsection*{From Section~\ref{sec:semi-aligned}: Semi-aligned Utilities\vspace{2mm}}

 \begin{lemma}
   \label{lem:linq}
   For a consumer optimal mixed profile $\mixsigma$ with filter
   strategy $\mixstrfil$ let $\tilde{\pi}_i(q(\mixstrfil))$ be the
   consumer's optimal information choice as give in lemma
   \ref{lem:cbr}. If $\qL<q(\mixstrfil)<\qH$,  then $(1-\tilde{\pi}_0(q(\mixstrfil))
   q(\mixstrfil)$,$(1-\tilde{\pi}_1(q(\mixstrfil)) (1-q(\mixstrfil))$ and
   $\tilde{\pi}_1(q(\mixstrfil) (1-q(\mixstrfil))$  are all linear in
   $q(\mixstrfil)$ and have no other dependence on $\pi_0$ and $\pi_1$.
 \end{lemma}

   \begin{proof}
     After tedious algebra by plugging in for $\mathcal{P}_c$ in
     equation \ref{lem:cbr}, it can be shown that
\begin{align*}
(1-\tilde{\pi}_0(q(\mixstrfil))q(\mixstrfil)
    &= \frac{e^{-c_2/\lambda}}{1-e^{-c_2/\lambda}}(q(\mixstrfil)-\qH),\\
\tilde{\pi}_1(q(\mixstrfil))(1-q(\mixstrfil))
    &= \frac{e^{-c_1/\lambda}}{e^{b/\lambda}-e^{-c_1/\lambda}}(q(\mixstrfil)-\qL),
\text{ and } \\
\tilde{\pi}_1(q(\mixstrfil))(1-q(\mixstrfil))
    &= \frac{e^{b/\lambda}}{e^{b/\lambda}-e^{-c_1/\lambda}}(\qH-q(\mixstrfil)),
\end{align*}
\noindent which are all linear in $q(\mixstrfil)$; $\pi_0$
     and $\pi_1$ only enter via $q(\mixstrfil)$.
   \end{proof}

 \begin{prop}\label{prop:nomixsemi}
$V_i(\mixsigma)\leq \max\rbr{V_i(\difprof), V_i(\fwdprof)}$
for any mixed equilibrium $\mixsigma$ and any
    $i\in\{\fil,\con\}$. Furthermore, if $\max(V_i(\difprof),
    V_i(\fwdprof))\neq V_i(\blkprof)$ the inequality is strict for any
    non-degenerate $\mixsigma$.
  \end{prop}

  \begin{proof}[Proof of Proposition \ref{prop:nomixsemi}]
    for the consumer, the proof follows directly from theorem
    \ref{thm:dvdp}.  
    It suffices to only consider the case where $\qdif<q_H$

   Any profile in which the filter blocks with positive probability at
   both information sets has payoffs bounded by $V(\blkprof)$ 
   therefore, it is only necessary to consider the case where the
   filter randomizes at one of its information sets and it sufficies to
   show that
   $\Vfil(\mixsigma)\leq max(\Vfil(\fwdprof), \Vfil(\difprof))$ when
   the filter always forwards content upon receiving a clean signal
   and randomizes otherwise.

   Let $\tilde{\pi}_i(q(\mixstrfil))$ be the consumer's unique optimal
   information strategy for filter strategy $\mixstrfil$ and
   $\beta_{\mixstrfil}$ be the unconditional probability the filter
   forwards.  Then: 
   \begin{align*}
     &\Vfil(\mixstrfil)
        = -c_1(1-q)(\pi_1\gamma + (1-\pi_1\gamma_1)) \;\;  \\
     &\quad+\beta_{\gamma}\big( q(\mixstrfil)(1-\tilde{\pi}_0(q(\mixstrfil)))(-c_2)\\
        &\qquad +(1-q(\mixstrfil))(\tilde{\pi}_1(q(\mixstrfil))(-c_1) + (1-\tilde{\pi}_1(q(\mixstrfil)))b)\big).
   \end{align*}
   Lemma \ref{lem:linq} establishes that the term inside the large
   parenthesis is linear.  Therefore, $\Vfil(\mixstrfil)$ is maximized
   either when $\gamma_0=1$ or $\gamma_0=0$ and $\gamma_1=0$ or when
   $q(\mixstrfil)=\qL$ (the other two pure profiles are weakly
   dominated as well as the profile where $q(\mixstrfil)=\qH$).
   Therefore, it suffices to show that
   $max(\Vfil(\difprof), \Vfil(\fwdprof))>\Vfil(\mixprof)$ when
   $q(\strdif)\leq q(\mixstrfil)\leq \qL$.  Suppose
   $q(\mixstrfil)<\qL$ and hold the consumer's strategy constant at
   accepting all content.  It is then trivial to show that the
   filter's payoffs are linear in $\gamma$.  If $V(\mixprof)$ is
   increasing in $\gamma$, then  $\Vfil(\mixprof)<\Vfil(\difprof)$.  If
   $V(\mixprof)$ is decreasing in $\gamma$, the
   $\Vfil(\fwdprof)>\Vfil(\mixprof)$, thus completing the proof.

 \end{proof}

\begin{proposition}[Existence of Equilibria]{\label{prop:semiebm}}
If $\qL<\qdif<\qH$, then $\difprof$ is an equilibrium if and only
       if $\frac{\pi_0(1-\pi_1)}{\pi_1(1-\pi_0)}>\Lambda$
 \end{proposition}

 \begin{proof}
The necessary and sufficient condition is that the filter does not have an
incentive to forward all content, holding the consumer's strategy
at its best response to $\difprof$.  This is
given by
\begin{align*}
  (1-q)\Vfil(\difprof)
    &>-c_2 q(1-\tilde{\pi}_0(q(\strdif))) \\
    &\quad -c_1(1-q)\tilde{\pi}_1(q(\strdif)) \\
    &\quad +b(1-q)(1-\tilde{\pi}_1(q(\strdif))).
  \end{align*}
  Substituting in the value for $\tilde{\pi}_0(q(\strdif))$ and
  $\tilde{\pi}_1(q(\strdif))$ and noting that
  \[ \frac{e^{b/\lambda}(1-e^{-c_2/\lambda})}{e^{-c_2/\lambda}
    (e^{b/\lambda}-e^{-c_1/\lambda})}=\frac{1-\qL}{\qL}\]
  yields the
  claimed result
\end{proof}

\begin{prewebconf}

  \begin{proof}[Proof of Proposition \ref{prop:semi-inefficiency}]
   Since $\qdif<\qH<q$, $V_i(\difprof)>V_i(\fwdprof)=V_i(\blkprof)$
   for $i \in \{\texttt{f}, \texttt{c}\}$.
   Furthermore, by lemma \ref{prop:nomixsemi},
   $V_i(\difprof)>V_i(\mixprof)$ for any purely mixed equilibrium.  This
   proves that $\difprof$ Pareto dominates any equilibrium.  Finally,
   proposition \ref{prop-semiebm}, establishes that existence of
   $\difprof$ as an equilibrium.
\end{proof}

\end{prewebconf}

 \begin{proof}[Proof of Theorem \ref{thm:semi-inef-gen}]
   Proposition  \ref{prop:semiebm} establishes that when
   $\mathcal{Q}(\pi_0, \pi_1)<\Lambda$, $\difprof$ is not an equilibrium but
   Proposition \ref{prop:nomixsemi} establishes that $\difprof$
   Pareto dominates any other equilibrium.  Therefore, when
   $\mathcal{Q}(\pi_0, \pi_1)<\Lambda$, any equilibrium is inefficient.
\end{proof}

 \begin{proof}[Proof of Theorem \ref{thm:misaligned-escape}]
Pareto efficiencyy follows from proposition \ref{prop:nomixsemi}.  The
positive VoTC follows by theorem \ref{thm:semi-VoT}.
   \end{proof}

 \begin{proof}[Proof of Theorem \ref{thm:semi-VoT}]

   The consumer \vot follows directly from theorem \ref{thm:dvdp}.
   For the filter, note that
    \begin{align*}
      V(\difprof)=  -c_1(1-q)(\pi_1) + \beta\big(\qdif(1-\tilde{\pi}_0(\qdif))(-c_2) + \nonumber \\
      (1-\qdif)(\tilde{\pi}_1(\qdif)(-c_1) + (1-\tilde{\pi}_1(\qdif))b)\big) \nonumber
    \end{align*}
    where lemma \ref{lem:linq} says that everything inside the
    parenthesis is linear in $\qdif$ and since
    $\qdif = \frac{(1-\pi_0)q}{\beta}$, the entire expression is linear
    in $\pi_0$ and $\pi_1$.  Taking derivatives is then straightforward.
  \end{proof}

\subsection*{From Section~\ref{sec:endog}: Endogenous Attacker\vspace{2mm}}

  \begin{lemma}
    \label{lem:qql}
    Let $\mixstr $ be a consumer optimal profile for some filter
    profile $\mixstrfil$.  Let $\rho_0^{-1}(x, \mixstr)$ be the value
    of $\rho_0$ such that $q(\rho_0, \mixstr)=x$.  Then The
    attacker's best response satisfies $\rho_0^{-1}(\qL, \mixstr)$.
    \end{lemma}
    \begin{proof}[Proof of Lemma \ref{lem:qql}]
      First, fix the filter's strategy at $\strfwd$. Note that the attacker
      choosing $\rho_0$ such $q(\rho_0, \strfwd)>\qH$
      is a weakly dominated strategy.  Also choosing $\rho_0$ such
      that $q(\rho_0, \mixstr)<\qL$ is not optimal since the consumer
      accepts all content and thus the attacker can do better by
      increasing $\rho$. For $q(\rho_0, \mixstr)\in(\qL,\qH)$, and
      letting $q(\rho_0) = \frac{\rho_0}{1+\rho_0}$ the
      attacker's payoff is
       \begin{align*}
         \Vatt &= \frac{q(\rho_0)}{1-q(\rho_0)}(1-\tilde{\pi}_0^*(\rho_0, \strfwd)) \nonumber \\
         & = e^{-c_2/\lambda} \big( \frac{1}{1-e^{-c_2/\lambda}} - \frac{e^{-c_1/\lambda}}{(e^{b/\lambda} + e^{-(c_1+c_2)/\lambda})(1-q(\rho_0))}\big)
       \end{align*}
       where $\tilde{\pi}_0^*(q, \strfwd)$ is given in Lemma
       \ref{lem:cbr}. The second fraction is negative but increasing
       in magnitude in $q$ and thus the optimal $\rho_0$ is to set
       $q(\rho_0, \fwdprof)=\qL$.  The case under $\difprof$ follows
       symmetrically, and since any mixed strategy is equivalent to a
       game with different values of $\pi_0$ and $\pi_1$, the result
       holds for any $\mixstrfil$.
     \end{proof}
     \begin{proof}[Proof of Theorem \ref{thm:eattack-zero}]
       Follows from Lemmas \ref{lem:qql} and \ref{lem:cbr}.
     \end{proof}

     \begin{lemma}
       \label{lem:fwdexists}
       There exists an equilibrium in which the filter chooses
       $\fwdprof$ if and only if $\Lambda>\frac{\pi_0}{\pi_1}$.
     \end{lemma}
     \begin{proof}
       By lemma \ref{lem:qql}, if $\fwdprof$ is an equilibrium,
       $q(\rho_0, \fwdprof)=\qL$.  Under $\difprof$, the
       consumer will accept all content, so
       \[ V(\rho_0^{-1}(\qL,
       \fwdprof), \fwdprof)> V(\rho_0^{-1}(\qL,
       \fwdprof), \difprof)\] if and only if
       \[ -\qL c_2 +(1-\qL)b > (1-\qL)(-\pi_1c_1 + (1-\pi_1)b)
       -\qL(1-\pi_0)c_2.\]  Algebra yields the claimed results.
       \end{proof}



     \begin{lemma}
       \label{lem:fwdbetter}
       \begin{align*}
       V(\rho_0^{-1}(\qL, \fwdprof), \fwdprof)>\displaystyle \max_{(\pi_0, \pi_1)\in
         (0,1)^2}V(\rho_0^{-1}(\qL, \difprof),
         \difprof) \nonumber
         \end{align*}
     \end{lemma}
     \begin{proof}
Simple equation manipulation shows
\begin{align*}
  &V(\rho_0^{-1}(\qL, \fwdprof), \fwdprof) - V(\rho_0^{-1}(\qL, \difprof)\\
    &\quad = \pi_1(b +c_1 -c_2 \frac{\qL}{1-\qL})
\end{align*}
The right-hand side is positive since consumer's best response when $q(\rho_0, \mixstrfil)=\qL$
implies that
\[-\qL c_2 + (1-\qL)b > -(1-\qL)c_1. \qedhere\]
\end{proof}

     \begin{proof}[Proof of Theorem \ref{thm:eattack}]
       Note that $\frac{\pi_0}{\pi_1}$ and
       $q(\rho_0^{-1}(\qL, \difprof), \fwdprof)$ are both increasing
       in $\pi_0$ and decreasing in $\pi_1$ (toward $\infty$ and $1$
       resp.).  Therefore, for sufficiently high values of $\pi_0$ and
       low values of $\pi_1$,
       $q(\rho_0^{-1}(\qL, \difprof), \fwdprof)>\qH$ and thus
       $(\rho_0^{-1}(\qL, \difprof)$ is the optimal equilibrium for
       the filter and consumer (since otherwise the consumer blocks
       all content).  By lemma \ref{lem:fwdexists}, $\fwdprof$ exists
       only when $\Lambda>\frac{\pi_0}{\pi_1}$.  Let $\pi'_0$ and
       $\pi'_1$ be such that $q(\rho_0^{-1}(\qL, \difprof),
       \fwdprof)>\qH$, and $\Lambda<\frac{\pi_0}{\pi_1}$.  Then
       anytime parameters change such that
       $\frac{\pi_0}{\pi_1}<\Lambda$ to $\pi'_0, \pi'_1$, by lemma
       \ref{lem:fwdbetter}, $\Vopt$ decreases.
       \end{proof}

\subsection*{Technical Details of the ``Unreasonable Strategy''}

\begin{prop}\label{prop:unreason2}
    The ``unreasonable''' strategy cannot be an equilibrium in which the consumer doesn't block all content.  
\end{prop}
\begin{proof}
Suppose the unreasonable strategy were played in an equilibrium in which the consumer accepted content with a positive probability.  . Then the filter would strictly profit by replacing its strategy with the mixed strategy $\pi_1  \strfwd + (1 - \pi_1) \strblk$. This yields strictly more utility because  under the unreasonable strategy, malicious content is forwarded at rate $\pi_0$ and genuine content is forwarded at rate $\pi_1$. Under the deviation, all content is forwarded at rate $\pi_1$. This deviation strictly increases content utility  and weakly reduces information costs.
\end{proof}
\begin{prop}\label{prop:nodeva}
In the aligned case, there is no profitable deviation from the socially optimal profile to the unreasonable profile.
\end{prop}
\begin{proof}
By proposition \ref{prop:unreason2} the payoffs of deviating to the unreasonable profile are upper bounded by the payoffs of a mixed strategy profile and proposition \ref{prop:aligned-best-eq} that payoff is less than the welfare maximizing profile.  
 \end{proof}

\begin{prop}
If $\difprof$ is an efficient equilibrium in the semi-aligned case, the filter cannot profit by deviating to the unreasonable profile.  
\end{prop}
\begin{proof}
The unreasonable strategy is bounded by the mixture of $\strfwd$ and $\strblk$ as in the proof of \ref{prop:unreason2}, which itself is bounded by the utility from $\strfwd$. Since by assumption $\difprof$ is an equilibrium (and therefore the filter has no incentive to switch to $\strfwd$) there is no incentive for the filter to deviate to the unreasonable profile.  
\end{proof}

Since the unreasonable profile cannot be an equilibrium and it doesn't nullify the equilibrium status of the equilibria we analyze, all of our results hold when also considering the unreasonable profile.

\end{document}